\documentclass[10pt,twoside,a4paper,reqno]{amsart}
\usepackage{amsfonts,amsthm,latexsym,amsmath,amssymb,amscd,epsf,latexsym}
\usepackage{epsfig}
\usepackage[dvips]{color}
\definecolor{LightGreen1}{rgb}{0.2,0.8,0.2}
\definecolor{DarkGray}{rgb}{0.4,0.4,0.4}
\definecolor{GhostWhite}{rgb}{0.97,0.97,1.}

\theoremstyle{plain}

\newtheorem{theo}           {Theorem}
\newtheorem{pro}            {Proposition}
\newtheorem{coro}           {Corollary}
\newtheorem{lemm}           {Lemma}
\newtheorem{conj}           {Conjecture}

\theoremstyle{definition}
\newtheorem{pr}              {Problem}
\newtheorem*{ack}            {Acknowledgements}

\newtheorem{nota}            {Notation}

\theoremstyle{remark}
\newtheorem{rem}             {Remark}

\theoremstyle{definition}

\newenvironment{theorem}{\begin{theo}}{\end{theo}}

\newenvironment{corollary}{\begin{coro}}{\end{coro}}
\newenvironment{lemma}{\begin{lemm}}{\end{lemm}}
\newenvironment{remark}{\begin{rem}}{\end{rem}}

\numberwithin{equation}{section}
\newtheorem{df}{\bf Definition}
\newtheorem{prop}{\bf Proposition}[section]

\newcommand \C {\mathcal C}

\newcommand \De {\Delta}

\newcommand {\bCP} {\mathbb {CP}}

\newcommand \Ga {\Gamma}
\newcommand \ga {\gamma}

\newcommand \bR {\mathbb R}
\newcommand \bC {\mathbb C}

\newcommand\al {\alpha}
\newcommand \be {\beta}

\newcommand{\Rea}{\mathbb {R}}      
\newcommand{\Cplx}{\mathbb {C}}     
\newcommand{\halmos}{\rule{5pt}{5pt}}

\newcommand \M {\mathcal M_{\le 2}^-}

\newcommand {\CC}{\mathcal C}

\newcommand{\la}{\lambda}

\begin{document}

\title[On spectral polynomials of the Heun equation. II.]
{On spectral polynomials of the Heun equation. II.}

\author[B.~Shapiro]{Boris Shapiro}
\address{Department of Mathematics, Stockholm University, SE-106 91
Stockholm,
      Sweden}
\email{shapiro@math.su.se}

\author[K.~Takemura]{Kouichi Takemura}
\address{Department of Mathematical Sciences, Yokohama City University, 22-2 Seto, Kanazawa-ku, Yokohama 236-0027, Japan.}
\email{takemura@yokohama-cu.ac.jp}

\author[M.~Tater]{Milo\v{s} Tater}
\address{Department of Theoretical Physics, Nuclear Physics Institute, 
Academy of Sciences, 250\,68 \v{R}e\v{z} near Prague, Czech
Republic}
\email{tater@ujf.cas.cz}

\date{\today}

\begin{abstract}
The  well-known {\em Heun equation}  has the form
$$
         \left\{Q(z)\frac
         {d^2}{dz^2}+P(z)\frac{d}{dz}+V(z)\right\}S(z)=0,
   $$
where $Q(z)$ is a cubic complex polynomial,  $P(z)$ and $V(z)$ are polynomials of degree at most $2$ and $1$ respectively.  One of the classical problems about the Heun equation suggested by E.~Heine and T.~Stieltjes in the late 19-th century is  for a given positive integer $n$ to find all possible polynomials $V(z)$ such that the above equation has a polynomial solution $S(z)$ of degree $n$. Below we prove a conjecture of the second author, see \cite {ShT} claiming that the union of the roots of such $V(z)$'s for a given $n$ tends when $n\to\infty$ to a certain compact connecting the three roots of $Q(z)$ which is given by a condition that a certain natural abelian integral is real-valued, see Theorem~\ref{takemura}. 

\end{abstract}

\maketitle

\section{Introduction and Main Results}\label{s:intro}

The classical Heun equation 
\begin{equation}\label{eq:Heun}
 \left\{Q(z)\frac
         {d^2}{dz^2}+P(z)\frac{d}{dz}+V(z)\right\}S(z)=0,
\end{equation}
where $Q(z)$ is a cubic polynomial, $P(z)$ is at most quadratic, and $V(z)$ is at most linear polynomials  was and still is an object of active study, see \cite{Heun}. Throughout this paper we always assume that $Q(z)$ is monic. The special case of \eqref{eq:Heun} when $P(z)=Q'(z)/2$ is widely known as the {\em classical Lam\'e} equation. Below we study one  aspect  of the Heun equation  suggested by E.~Heine and T.~Stieltjes, see \cite{He}, \cite {St},  and \cite{WW}, ch.~23. 

\begin{pr}[Heine-Stieltjes]\label{prob:He-St} 
For a given pair of polynomials $Q(z)$ and $P(z)$ as above and a positive integer $n$ find all polynomials $V(z)$ such that \eqref{eq:Heun} has a polynomial solution $S(z)$ of degree $n$.
\end{pr}  

Polynomials $V(z)\;$ (resp. $S(z)$)  are usually referred to as {\em Van Vleck\phantom{x}}  (resp.  {\em Stietljes}, or sometimes, {\em Heine-Stieltjes}) Ê polynomials. Already Heine and Stieltjes knew that for a generic pair $(Q(z),P(z))$ and any positive integer  $n$ there exist exactly $n+1$ such distinct Van Vleck polynomials $V(z)$.  Moreover in the case of the Lam\'e equation when one additionally assumes that the polynomial $Q(z)$ has three distinct real roots $\al<\be<\ga$ resp. Stieltjes  was  able to prove that the roots of any $V(z)$ and $S(z)$ belong to the interval $(\al,\ga)$ and that for a given $n$ the $(n+1)$ existing Stieltjes polynomials are distinguished by how many of their roots lie in the interval $(\al,\be)$ (the remaining roots lie in  
the interval $(\be,\ga)$, see \cite {WW}, ch. 23, section 46.)   Some further information of asymptotic character can be found in \cite{BSh} and \cite{MFS}. 

For a general Heun equation no essential results about the location of the roots of Van Vleck and Stieltjes polynomials seems to be previously known. One of the few exceptions is a classical 
proposition of  P\'olya, \cite{Po} claiming that if the rational function $\frac{P(z)}{Q(z)}$ has all positive residues then any root of any $V(z)$ as above and of any $S(z)$ as above lies within   $\Delta_Q$ where $\Delta_Q$ is the convex hull of the set of all three roots of $Q(z)$. 

The next  statement is a specialization of the main  result of \cite{Sh} in the case of the Heun equation.

 \begin{theorem}\label{th:my} For any cubic polynomial $Q(z)$  and any polynomial $P(z)$ of degree at most $2$ one has that 
 \begin{enumerate}
 \item there exists $N$ such that for any $n\ge N$ there exist exactly $n+1$ linear polynomials $V(z)$ counted with appropriate multiplicity such that \eqref{eq:Heun} has a polynomial solution $S(z)$ of degree exactly $n$;
 \item for any $\epsilon >0$ there exists $N_\epsilon$ such that for  any $n\ge N_\epsilon$ any root of any $V(z)$ having  $S(z)$ of degree $n$ as well as any root of this $S(z)$ lie in the $\epsilon$-neighborhood of $\Delta_Q$.
 \end{enumerate}
  \end{theorem}

   Thus we can introduce the set $\mathcal V_n$ consisting of polynomials $V(z)$ giving a polynomial solution $S(z)$ of \eqref{eq:Heun} of degree $n$; each such $V(z)$ appearing in $\mathcal V_n$ the number of times equal to its multiplicity. (The exact definition of multiplicity of $V(z)$ is rather lengthy and is omitted here.   An interested reader is recommended to consult \cite {Sh} for details.) Then by the above results  the set $\mathcal V_n$ will contain exactly $n+1$ linear polynomials for all sufficiently large $n$. It will be convenient to introduce a  sequence $\{Sp_n(\lambda)\}$ of {\em spectral polynomials} where the $n$-th spectral polynomial is defined by 
$$Sp_n(\lambda)=\prod_{j=1}^{n+1}(\lambda-t_{n,j}),$$ 
$t_{n,j}$ being the unique root of the $j$-th polynomial in $\mathcal V_n$  in any fixed ordering.  Notice that $Sp_n(\lambda)$  is well-defined for all sufficiently large  $n$.

Associate to $Sp_n(\lambda)$ the finite measure
$$\mu_n=\frac{1}{n+1}\sum_{j=1}^{n+1}{\delta(z-t_{n,j})},$$
where $\delta(z-a)$ is the Dirac measure supported at $a$.
The measure $\mu_n$  obtained in this way is clearly a real probability measure which 
one usually refers to as the {\em root-counting measure} of the polynomial $Sp_n(\lambda)$. 

Our main question is as follows.

\begin{pr}\label{prob:asymp}  Does the sequence $\{\mu_n\}$ converge (in the weak sense) to some limiting measure $\mu$?  If the convergence takes place describe the limiting measure $\mu$?
\end{pr}

Below we answer both parts of this question, see Theorem~\ref{takemura}. With an essential  contribution of the second author we were able to prove the existence of $\mu$ and to find the following elegant  description of its support. 

Denote the three roots of $Q(z)$ by $a_1,a_2,a_3$. For $i\in\{1,2,3\}$  consider the curve $\gamma_i$ given as the set of all $b\in \bC$ satisfying the relation:  
\begin{equation}\label{eq:ba123Rea}
\gamma_i:\quad\int_{a_j}^{a_k}\sqrt{\frac{b-t}{(t-a_1)(t-a_2)(t-a_3)}}dt\in \bR,
\end{equation}
here $j$ and $k$ are the remaining two indices in $\{1,2,3\}$ in any order and the integration is taken over the straight interval connecting $a_j$ and $a_k$. One can see that $a_i$ belongs to $\gamma_i$ and  show that these three curves connect all $a_i$'s with the unique common point $b_0$ lying within $\Delta_Q$, see Lemma~\ref{lm:inter}, \S 4. Take the segment of $\gamma_i$ connecting $a_i$ with the common intersection point $b_0$  and denote this segment by $\Gamma_i$, see Fig.~1 and Fig.~5.  Finally,  denote the union of these three segments by $\Gamma_Q$. 

Our first  result is as follows. 

\begin{theorem} \label{takemura} 
(i) For any equation \eqref{eq:Heun} the sequence $\{\mu_n\}$ of the root-counting measures of its spectral polynomials converges  to a probability measure $\mu$ depending only on the leading coefficient  $Q(z)$; 

\noindent 
(ii) The support of the limiting root-counting measure $\mu$ coincides with $\Gamma_Q$. 
\end{theorem}

\begin{rem} Knowing the support of $\mu$ it is also possible to define its density along the support using  the linear differential equation satisfied by its Cauchy transform, see Theorem 4 of \cite {ShT}. In case
when $Q(z)$ has all real zeros the density is explicitly  given in \cite {BSh}. 
\end{rem}

\begin{figure}[!htb]
\centerline{\hbox{\epsfysize=4cm\epsfbox{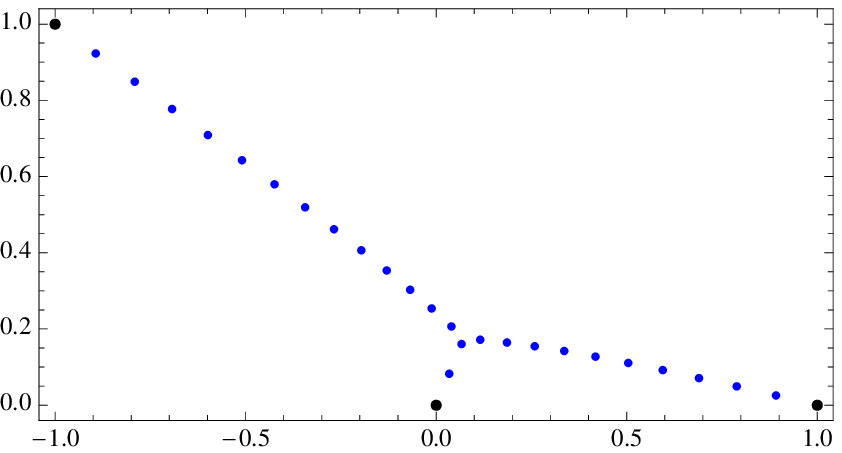}}}
{Figure 1. The roots of the spectral polynomial $Sp_{24}(\lambda)$ for the  equation 
$Q(z)S^{\prime\prime}(z)+V(z)S(z)=0$ with
$Q(z)=z(z-1)(z-1+I)$.}
\label{fig1}
\end{figure}


\medskip



An essential role in the proof of Theorem~\ref{takemura} is played by the description of the behavior of the Stokes lines of \eqref{eq:Heun}.  Important contribution also comes from a generalization of  the technique of \cite {KvA}. In particular,   in \cite{ShT} using the latter technique Êwe were able to find an additional probability measure which is easily described and from which the measure  $\mu$   is obtained by the inverse  balayage, i.e. the support of $\mu$ will be contained in the support of the measure which we construct and they have the same logarithmic potential outside the support of the latter  one. This measure is  uniquely determined by  the choice of a root of $Q(z)$ and thus we have in fact constructed three different measures having the same measure $\mu$ as their inverse balayage.    

Our second result describes the asymptotic behavior of Stieltjes polynomials of increasing degrees  when the 	sequence of  their (normalized) Van Vleck polynomials has a limit.  This result is a special case of a more general statement of \cite{HSh} but we explain below in much more details the interaction  of the limiting  measure  with the appropriate rational Strebel differential. 

Namely, for a  Heun equation (\ref{eq:Heun}) take  any sequence $\{S_{n,i_{n}}(z)\},\;\deg  S_{n,i_{n}}(z)=n$ of its  Stieltjes  polynomials  such that the sequence of normalized Van Vleck polynomials $\{\widetilde  V_{n,i_{n}}(z)\}$ converges   to some  monic linear polynomial $\widetilde V(z)$.  Here by normalization we mean the division by the leading coefficient, i.e each $\widetilde  V_{n,i_{n}}(z)$ is the monic polynomial proportional to $V_{n,i_{n}}(z)$. Notice that since each $\widetilde  V_{n,i_{n}}(z)$ is linear for all sufficiently large $n$ then  the existence of the limiting polynomial  $\widetilde V(z)$ is the same as the existence of the limit of the sequence of (unique) roots $\{b_{n,i_n}\}$ of $\{V_{n,i_{n}}(z)\}$. Part 2 of Theorem~\ref{th:my} guarantees the existence of plenty of such converging sequences and Theorem~\ref{takemura} Êclaims that the limit $\tilde b$ of these roots must necessarily belong to $\Ga_Q$. 

Finally, denote by  $\nu_{n,i_{n}}$ of the root-counting measure of the corresponding Stieltjes polynomial $S_{n,j_{n}}(z).$

 \medskip
     \begin{theorem}\label{th:higRull}  In the above notation the sequence $\{\nu_{n,i_{n}}\}$ of the root-counting measures of the corresponding Stieltjes polynomials $\{S_{n,j_{n}}(z)\}$ weakly converges to the unique probability measure $\nu_{\widetilde V}$ whose Cauchy  transform $\CC_{\widetilde V}(z)$ satisfies almost everywhere  in $\bC$ the equation  
$$\CC_{\widetilde V}^2(z)=\frac {\widetilde V(z)}{Q(z)}.$$
\end{theorem}

Typical behavior of the roots of $\{S_{n,j_{n}}(z)\}$ is illustrated on Fig.~3 below. 
Recall that   the Cauchy transform $\C_\nu(z)$ and the logarithmic potential $u_\nu(z)$ of a (complex-valued) measure $\nu$ supported in $\bC$ are by definition given by: 
$$\C_\nu(z)=\int_{\bC}\frac{d\nu(\xi)}{z-\xi}\quad\text{    and    }\quad u_\nu(z)=\int_{\bC}\log|z-\xi|{d\nu(\xi)}.$$
Obviously, $\C_\mu(z)$ is analytic outside the support of $\mu$ and
has a number of important properties, e.g. that
$\mu=\frac{1}{\pi}\frac{\C_\mu(z)}{\partial \bar z}$ where the derivative is understood in
the distributional sense.  Detailed information about Cauchy transforms
can be found in e.g. \cite{Ga}.


To formulate our further results we need to recall some information about quadratic differentials.

\begin{df}
A (meromorphic) quadratic differential $\Psi$ on a compact  orientable Riemann surface $Y$ is
a (meromorphic) section of the tensor square
$(T^*_{\mathbb{C}}Y)^{\otimes2}$ of the holomorphic cotangent
bundle $(T^*_{\mathbb{C}}Y)$. The zeros and the poles of $\Psi$
constitute the set of \emph{singular points} of $\Psi$ denoted by
$Sing_{\Psi}$. (Nonsingular points of $\Psi$ are usually called {\em
regular}.)
\end{df}

Obviously, if $\Psi$ is locally represented in two intersecting charts by $h(z)dz^2$ and by
$\tilde{h}(\tilde{z})d\tilde{z}^2$  resp. with a
transition function $\tilde{z}(z)$, then 
$h(z)=\tilde{h}(\tilde{z})\left(\frac{d\tilde{z}}{dz}\right)^2.$
Any quadratic differential induces a canonical metric on its Riemann
surface, whose length element  in local coordinates is given by
$$|dw|=|h(z)|^{\frac{1}{2}}|dz|.$$

The above canonical  metric $|dw|=|h(z)|^{\frac{1}{2}}|dz|$ on $Y$ is closely related to  two distinguished line fields  given by the condition that
$h(z)dz^2$ is either positive or negative.  The first field is given by 
$h(z)dz^2>0$ and its integral curves are called \emph{horizontal
trajectories} of $\Psi$, while the second field is given by $h(z)dz^2<0$ and
its integral curves  are called \emph{vertical trajectories} of $\Psi$. In what follows we
will mostly use horizontal trajectories of rational
quadratic differentials  and reserve the term 
\emph{trajectories} for the horizontal ones. In case we need vertical trajectories as in the \S~4 we will mention this explicitly. 

Since we only consider rational quadratic differentials  here then any such  quadratic differential $\Psi$ will be
given  in $\bC$ by $R(z)dz^2$, where $R(z)$ is a complex-valued rational function. (To study the 
behavior of $\Psi$ at infinity one makes the variable  change 
$\tilde{z}=\frac{1}{z}$.)\\

Trajectories of $\Psi$ can be naturally  parameterized by their 
arclength. In fact, in a neighborhood of a regular point $z_0$ on
$\bC$
 one can introduce a local coordinate called {\em canonical}Ê or  {\em distinguished parameter}
and  given by
\[w(z)
:=\int_{z_0}^{z}\sqrt{R(\xi)}d\xi.\]
One can easily check that 
$dw^2={R(z)}dz^2$ 
implying that  horizontal trajectories in the $z$-plane correspond to 
horizontal straight lines in the $w$-plane, i.e they are defined by the condition $\Im\, w=const$.

In what follows quadratic differentials which we encounter will be mostly Strebel, i.e. with almost all closed trajectories,  see the exact definition below. They are known to have the property that their poles are at most quadratic and, additionally,   the coefficient at the leading term of any such quadratic pole is negative. Introduce the class $\M$ of meromorphic  quadratic   on a Riemann surface $Y$ satisfying  the above restrictions, i.e. their poles are at most of order $2$ and at each such pole the leading coefficient is negative. 

\begin{df}
A trajectory of $\Psi\in \M$  is called  \emph{singular} if there exists a singular point of $\Psi$ belonging to its closure. 
\end{df}

\begin{df}
A non-singular trajectory $\gamma_{z_0}(t)$ of $\Psi\in \M$  is called \emph{closed} if $\exists \ T>0$ such that
$\gamma_{z_0}(t+T)=\gamma_{z_0}(t)$ for all $t\in\mathbb{R}$. The
least such $T$ is called the \emph{period} of $\gamma_{z_0}$. 
\end{df}

\begin{df}
A quadratic differential $\Psi$ on a Riemann surface $Y$ is called
\emph{Strebel} (or  {\em Jenkins-Strebel}) if the set of its 
closed trajectories covers the surface up to a set of Lebesgue measure 
 zero.
\end{df}

The following statement claiming that each Strebel differential of a compact Riemann surface automatically belongs to $\M$ can be derived from results of Ch. 3, \cite{Str}.  

\begin{lemma}
\label{lemma1} If a quadratic   differential $\Psi$
 is Strebel, then it has no poles of order
greater than 2. If it has a pole of order 2, then the coefficient at
leading term  of $\Psi$ at this pole is negative.
\end{lemma}

\begin{df}ÊFor a given quadratic differential $\Psi\in \M$ on  a compact surface  $Y$  denote by 
$K_\Psi\subset Y$ the union of all its singular trajectories and singular points. 
\end{df}

In many situations  a quadratic differential $\Psi$ on  a compact surface  $Y$ is Strebel if and only if the set  $K_\Psi$ is compact, see e.g. Theorem 20.1 of  \cite{Str}. Unfortunately we were unable to find an appropriate result for rational quadratic differentials in the literature (although it is apparently known) and include a sketch of its proof  as Lemma~\ref{lemma2} below.  Our next result relates Strebel differentials to real-valued measures in the considered  situation. 

\medskip
\begin{theorem}\label{th:charac} Let $U_1(z)$ and $U_2(z)$ be arbitrary  monic complex polynomials
with $\deg U_2 -\deg U_1=2$. Then 
\begin{enumerate}
\item
the  rational quadratic differential
$\Psi=-\frac{U_1(z)}{U_2(z)}dz^2$ on $\bCP^1$ is Strebel if and only if   
there exists a real and compactly supported in $\mathbb{C}$ measure
$\mu$ of total mass $1$ (i.e. $\int_\mathbb{C}d\mu=1$)
 whose Cauchy transform $C_{\mu}$ satisfies a.e. in $\bC$ the equation: 
\begin{equation}\label{eq:imp}
C_{\mu}^2(z)=\frac{U_1(z)}{U_2(z)}.
\end{equation}

\item  for any $\Psi$ as in (1) there exists exactly $2^{d-1}$  real measures whose Cauchy transforms  satisfy \eqref{eq:imp} a.e. and whose support is contained in $K_\psi$.    Here $d$ is the total number of connected components in $\bCP^1\setminus K_\Psi$ (including the infinite component, i.e. the one containing $\infty$). 

\end{enumerate} 
\end{theorem}

\begin{rem}
The above theorem is illustrated on Fig.~2. Notice that we are {\bf not}  assuming here that $\mu$ as above is a positive measure {\bf but} only real. (Such measures are sometimes called {\it signed}.) 
\end{rem}

\begin{rem} Notice that if we do not require  the support of a real measure whose Cauchy transform satisfy \eqref{eq:imp} a.e.  to be contained in $K_\psi$ then there exists plenty of such  measures. In particular, their support can contain an arbitrary finite number of distinct closed trajectories of the quadratic differential $\Psi=-\frac{U_1(z)}{U_2(z)}dz^2$ near infinity. However, as we will see late at most one such measure can be positive. 
\end{rem}

\begin{center}
\begin{picture}(440,270)(0,0)
\put(170,240){\line(1,0){60}}
\put(170,240){\circle*{3}}
\put(230,240){\circle*{3}}

\qbezier(200,275)(250,270)(230,240)
\qbezier(200,275)(150,270)(170,240)
\qbezier(200,205)(250,210)(230,240)
\qbezier(200,205)(150,210)(170,240)

\put(180,255){\circle*{3}}
\put(220,255){\circle*{3}}
\qbezier(180,255)(200,265)(220,255)

\put(180,225){\circle*{3}}
\put(220,225){\circle*{3}}
\qbezier(180,225)(200,215)(220,225)

\put(50,160){Figure 2. The union of singular trajectories $K_\Psi$ and 4 real measures}

\put(20,90){\line(1,0){60}}
\put(20,90){\circle*{3}}
\put(80,90){\circle*{3}}
\put(55,90){$\ominus$}

\put (50,130) {\vector(0,1){6}}
\put (50,115) {\vector(0,1){6}}
\put (50,105) {\vector(0,-1){6}}

\put (50,75) {\vector(0,1){6}}
\put (50,65) {\vector(0,-1){6}}
\put (50,50) {\vector(0,-1){6}}

\put(30,105){\circle*{3}}
\put(70,105){\circle*{3}}
\qbezier(30,105)(50,115)(70,105)
\put(55,108){$\oplus$}

\put(30,75){\circle*{3}}
\put(70,75){\circle*{3}}
\qbezier(30,75)(50,65)(70,75)
\put(55,70){$\oplus$}

\put(15,10){1st measure}

\put(120,90){\circle*{3}}
\put(180,90){\circle*{3}}

\qbezier(150,125)(200,120)(180,90)
\put(155,108){$\ominus$}
\qbezier(150,125)(100,120)(120,90)
\put(155,70){$\oplus$}

\put (150,130) {\vector(0,1){6}}
\put (150,120) {\vector(0,-1){6}}
\put (150,100) {\vector(0,1){6}}

\put (150,65) {\vector(0,-1){6}}
\put (150,50) {\vector(0,-1){6}}
\put (150,75) {\vector(0,1){6}}

\put(130,105){\circle*{3}}
\put(170,105){\circle*{3}}
\qbezier(130,105)(150,115)(170,105)

\put(130,75){\circle*{3}}
\put(170,75){\circle*{3}}
\qbezier(130,75)(150,65)(170,75)
\put(155,120){$\oplus$}

\put(115,10){2nd measure}

\put(220,90){\circle*{3}}
\put(280,90){\circle*{3}}

\qbezier(250,55)(300,60)(280,90)
\put(255,108){$\oplus$}
\qbezier(250,55)(200,60)(220,90)
\put(255,70){$\ominus$}

\put (250,130) {\vector(0,1){6}}
\put (250,115) {\vector(0,1){6}}
\put (250,105) {\vector(0,-1){6}}

\put (250,78) {\vector(0,-1){6}}
\put (250,62) {\vector(0,1){6}}
\put (250,50) {\vector(0,-1){6}}

\put(230,105){\circle*{3}}
\put(270,105){\circle*{3}}
\qbezier(230,105)(250,115)(270,105)

\put(230,75){\circle*{3}}
\put(270,75){\circle*{3}}
\qbezier(230,75)(250,65)(270,75)
\put(255,55){$\oplus$}

\put(215,10){3rd measure}

\put(320,90){\line(1,0){60}}
\put(320,90){\circle*{3}}
\put(380,90){\circle*{3}}

\qbezier(350,125)(400,120)(380,90)
\put(355,108){$\ominus$}
\qbezier(350,125)(300,120)(320,90)
\put(355,70){$\ominus$}
\qbezier(350,55)(400,60)(380,90)
\put(355,123){$\oplus$}
\qbezier(350,55)(300,60)(320,90)

\put (350,130) {\vector(0,1){6}}
\put (350,118) {\vector(0,-1){6}}
\put (350,102) {\vector(0,1){6}}

\put (350,78) {\vector(0,-1){6}}
\put (350,62) {\vector(0,1){6}}
\put (350,50) {\vector(0,-1){6}}

\put(330,105){\circle*{3}}
\put(370,105){\circle*{3}}
\qbezier(330,105)(350,115)(370,105)
\put(355,55){$\oplus$}

\put(330,75){\circle*{3}}
\put(370,75){\circle*{3}}
\qbezier(330,75)(350,65)(370,75)

\put(315,10){4th measure}

\label{figNew}
\end{picture}
\end{center}

\noindent
{\em Explanation to Fig.~2}. The picture on top shows the  union $K_\Psi$ of singular trajectories and singular points of an appropriate  Strebel rational quadratic differential  $\Psi=-\frac{U_1(z)}{U_2(z)}dz^2$ with  $2$ simple zeros and $4$ simple poles. (It is well-known that any such picture is realizable by a suitable  Strebel differential.) The set  $\bC\setminus K_\Psi$ has two  bounded connected component and one unbounded. Choosing a branch of square root $\sqrt {R(z)}$ in each of the two connected components we define and unique real measure supported on (a part of) $K_\Psi$. Four pictures in the second row show the actual support of these four measures where $\oplus$ and $\ominus$ indicate the sign of the measure on the corresponding part of its support. Finally, arrows show the direction of the gradient of the logarithmic potential  of the corresponding measures in  respective components of $\bC\setminus K_\Psi$. They determine which singular trajectories are present in the support and which are not. (See details in \S~2.) 

 Concerning  positive measures we claim the following. 

\begin{theorem} \label{pr:positive} In the notation of Theorem~\ref{th:charac} for any Strebel diiferential $\Psi=-\frac{U_1(z)}{U_2(z)}dz^2$ on $\bCP^1$  there exists at most one positive measure 
satisfying \eqref{eq:imp} a.e. in $\bC$. Its support necessarily belongs to $K_\Psi$, and, therefore,  among  $2^{d-1}$ real measures  described in Theorem~\ref{th:charac} at most one is positive. 
\end{theorem}Ê

In \S~\ref{sec:2} we provide an exact criterion of the existence of a positive measure in terms of rather simple topological properties of $K_\Psi$.  In particular, in the case shown on Fig.~2  no such positive measure exists as one can see from the signs on different parts of the support.  It seems very likely that results similar to Theorem~\ref{th:charac} and \ref{pr:positive} can be proved for Riemann surfaces of positive genera as well. 

The latter theorem together with Theorem~\ref{th:higRull} imply the following.  

\begin{corollary}Ê\label{cor:main} Under the above assumptions  one has that 
     the quadratic differential $\Phi=-\frac {\widetilde V(z)}{Q(z)}dz^2$ is Strebel where ${\widetilde V(z)}$ is the limit of some sequence of normalized Van Vleck polynomials for \eqref{eq:Heun}.
     \end{corollary}
     
     Finally, detailed study of quadratic differentials in \S 4 results in the following statement. 

\begin{pro}\label{pr:struc}
  For $\tilde b\neq b_0$ the support $\frak {S}_{\widetilde V}$ of  $\nu_{\widetilde V}$ consists of two singular trajectories of $\Phi$ - one connecting two roots among $a_1,a_2,a_3$ and one connecting the remaining root with $\tilde b$. For $\tilde b=b_0$ this support consists of three singular trajectories connecting $a_1,a_2,a_3$ to $b_0$, see Fig.~5 and 6.  
  \end{pro}
     

\begin{figure}[!htb]

\centerline{\hbox{\epsfysize=1.6cm\epsfbox{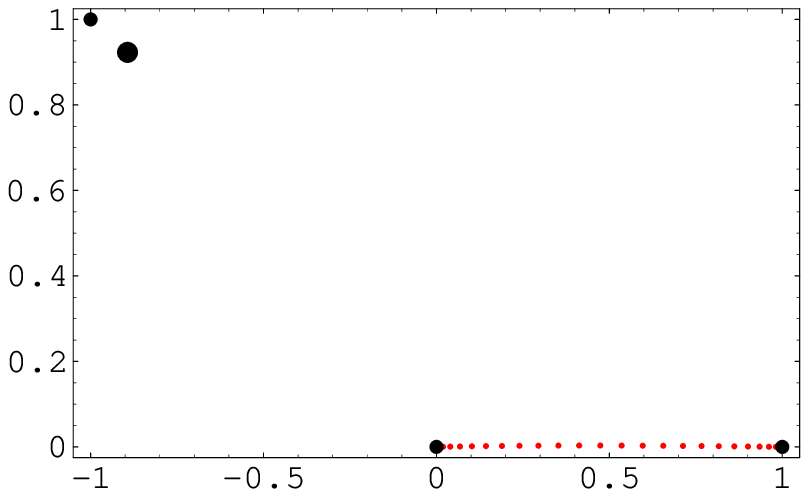}}
\hskip0.5cm\hbox{\epsfysize=1.6cm\epsfbox{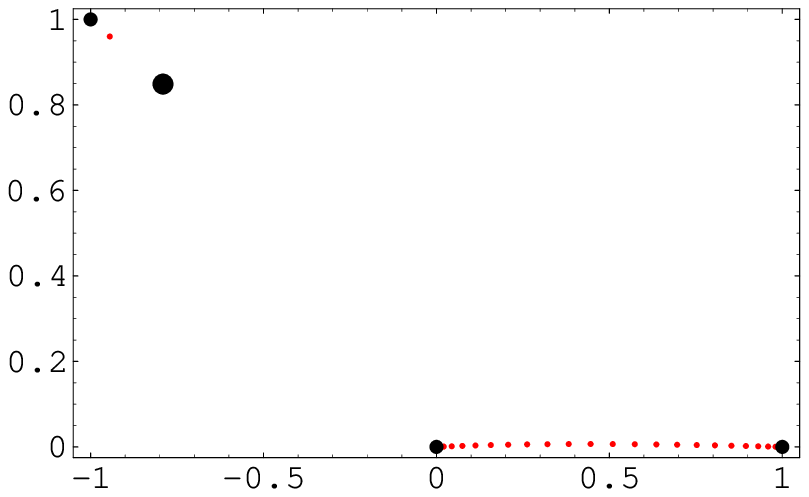}}
\hskip0.5cm\hbox{\epsfysize=1.6cm\epsfbox{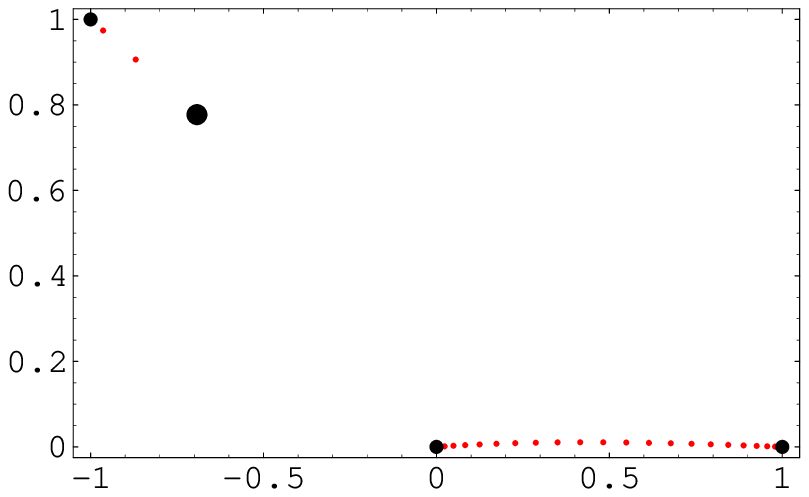}}
\hskip0.5cm\hbox{\epsfysize=1.6cm\epsfbox{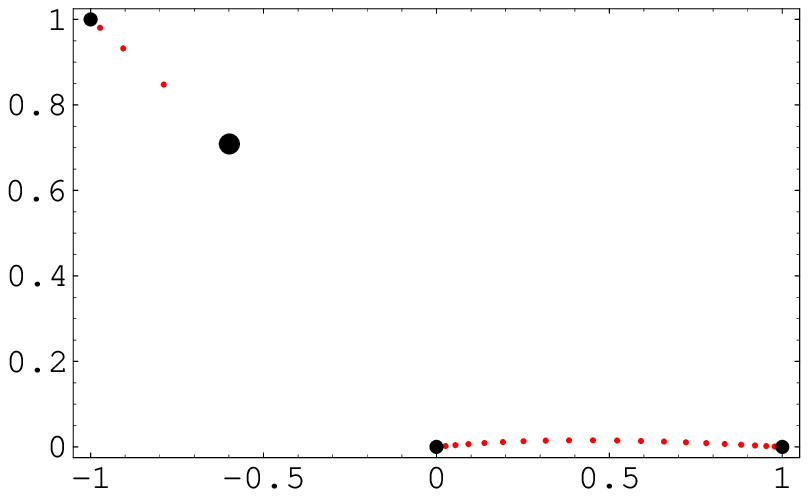}}
\hskip0.5cm\hbox{\epsfysize=1.6cm\epsfbox{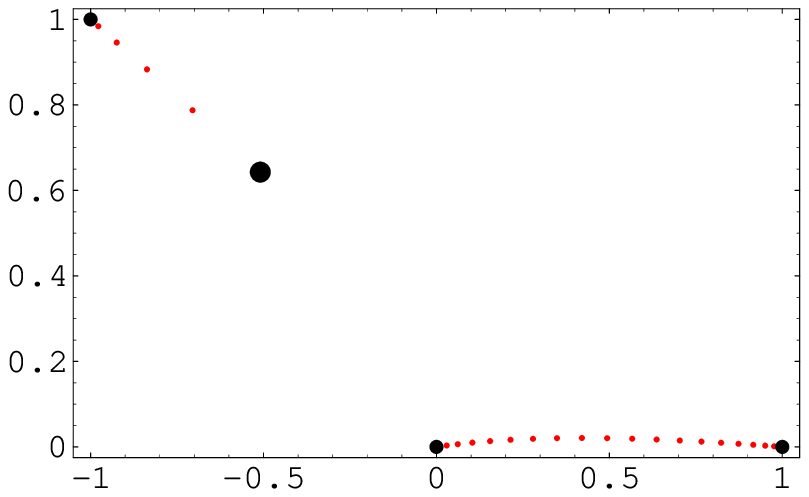}}
}

\centerline{\hbox{\epsfysize=1.6cm\epsfbox{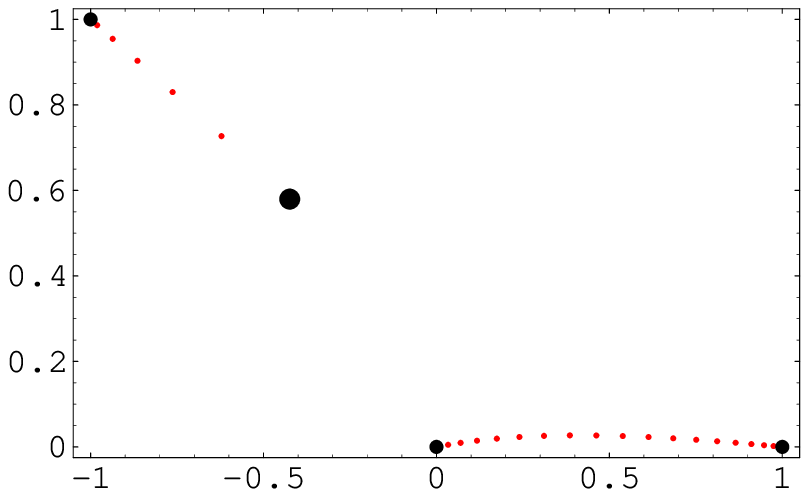}}
\hskip0.5cm\hbox{\epsfysize=1.6cm\epsfbox{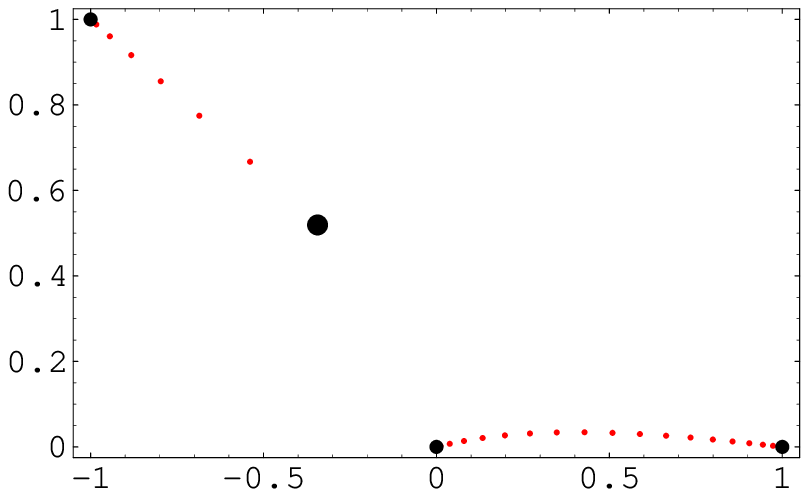}}
\hskip0.5cm\hbox{\epsfysize=1.6cm\epsfbox{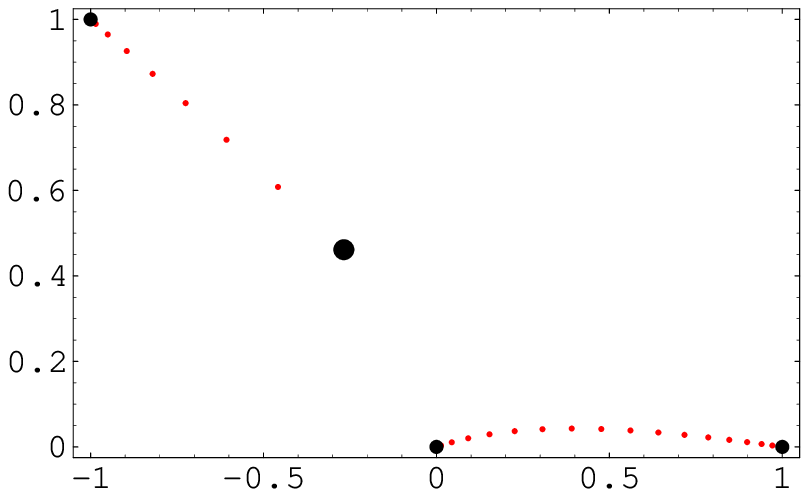}}
\hskip0.5cm\hbox{\epsfysize=1.6cm\epsfbox{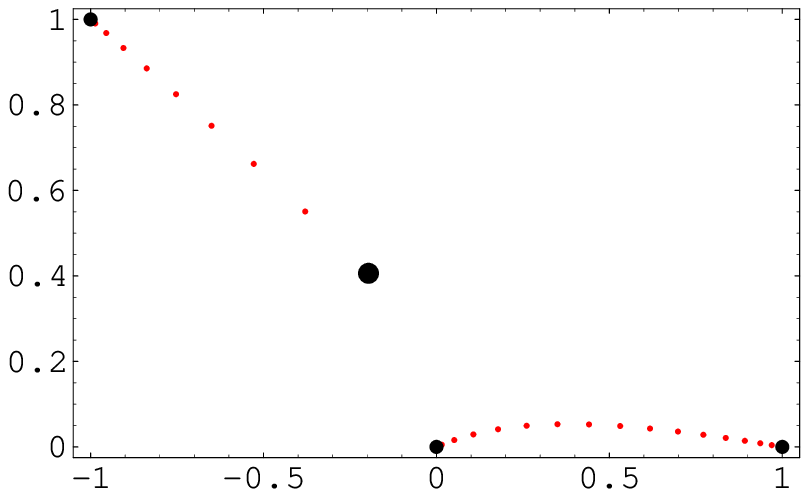}}
\hskip0.5cm\hbox{\epsfysize=1.6cm\epsfbox{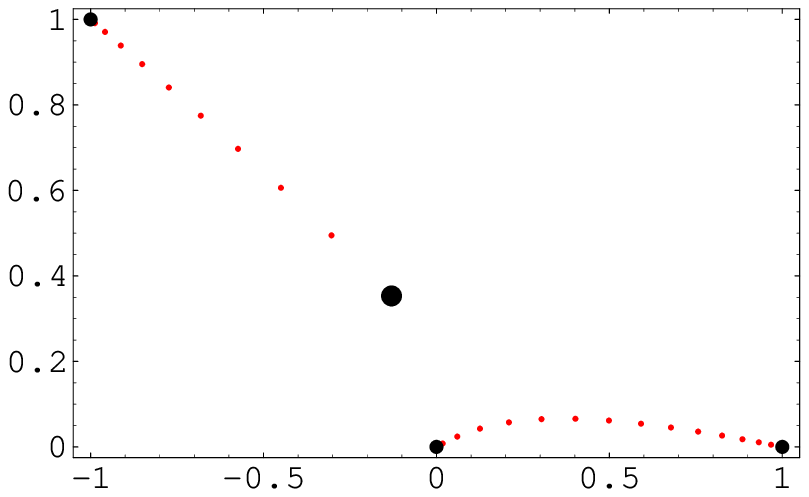}}
}

\centerline{\hbox{\epsfysize=1.6cm\epsfbox{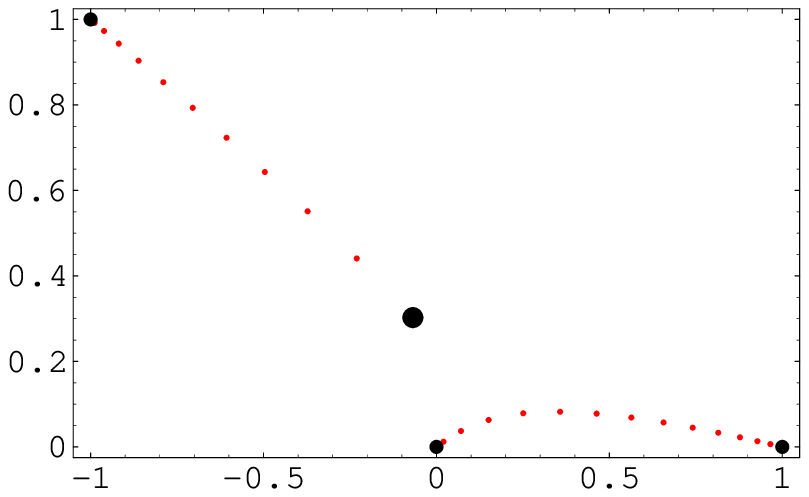}}
\hskip0.5cm\hbox{\epsfysize=1.6cm\epsfbox{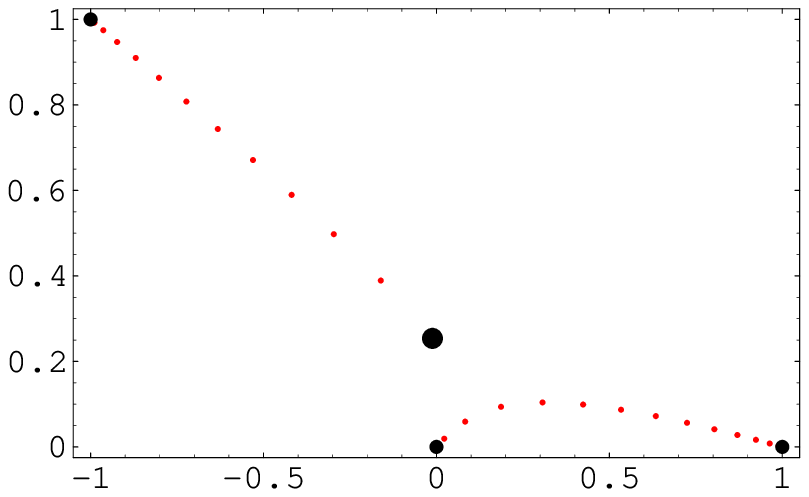}}
\hskip0.5cm\hbox{\epsfysize=1.6cm\epsfbox{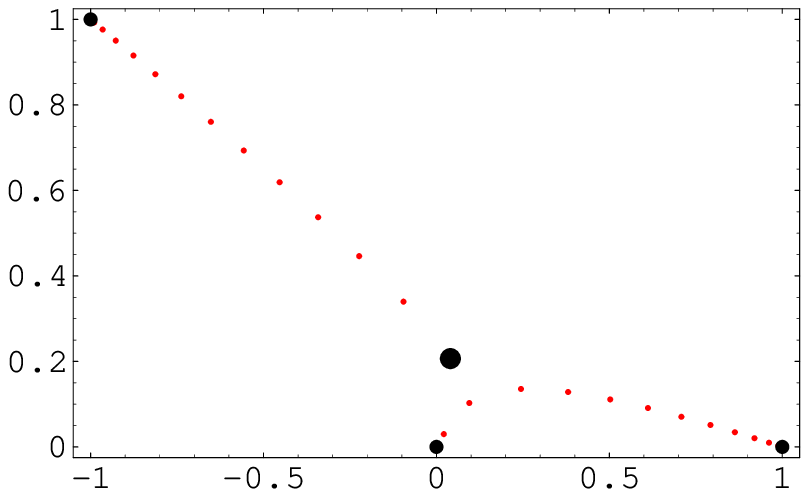}}
\hskip0.5cm\hbox{\epsfysize=1.6cm\epsfbox{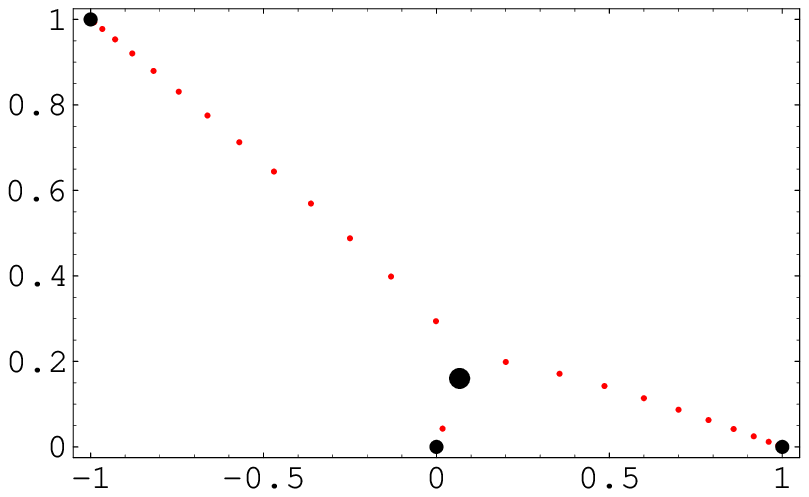}}
\hskip0.5cm\hbox{\epsfysize=1.6cm\epsfbox{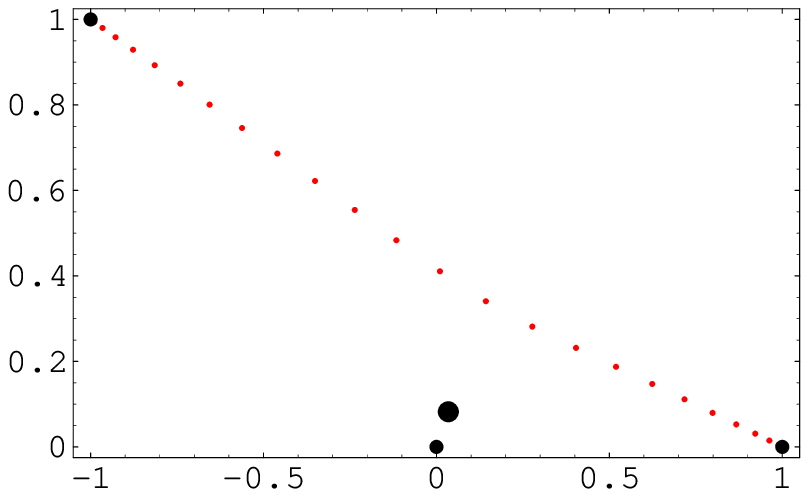}}
}

\centerline{\hbox{\epsfysize=1.6cm\epsfbox{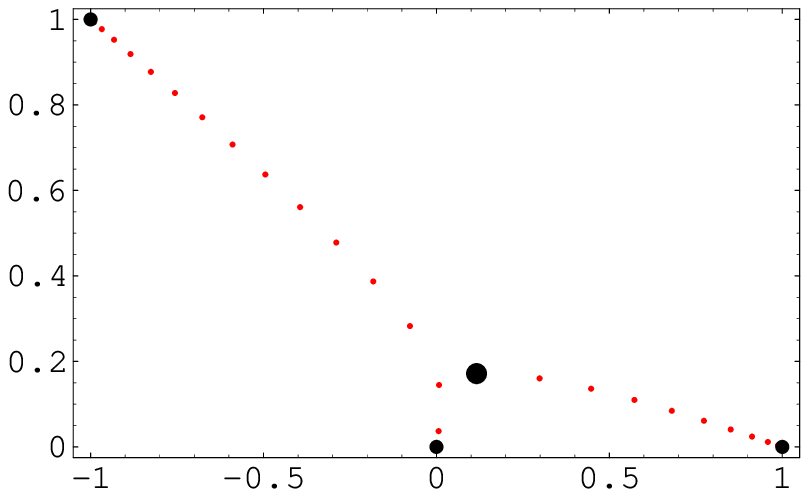}}
\hskip0.5cm\hbox{\epsfysize=1.6cm\epsfbox{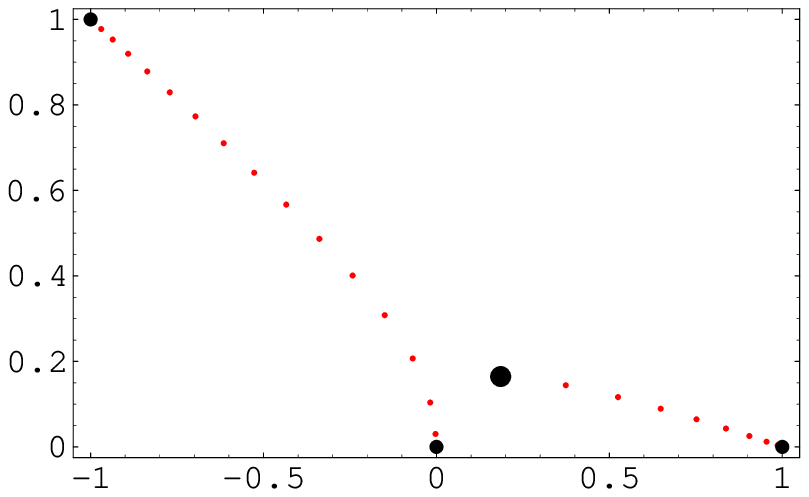}}
\hskip0.5cm\hbox{\epsfysize=1.6cm\epsfbox{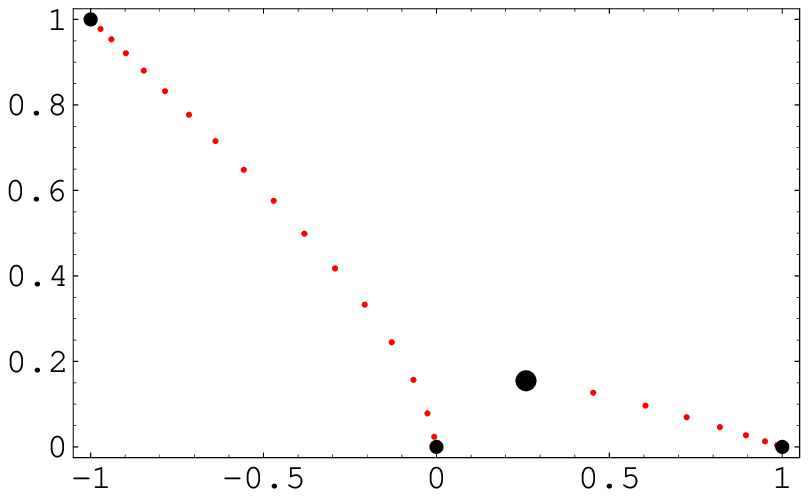}}
\hskip0.5cm\hbox{\epsfysize=1.6cm\epsfbox{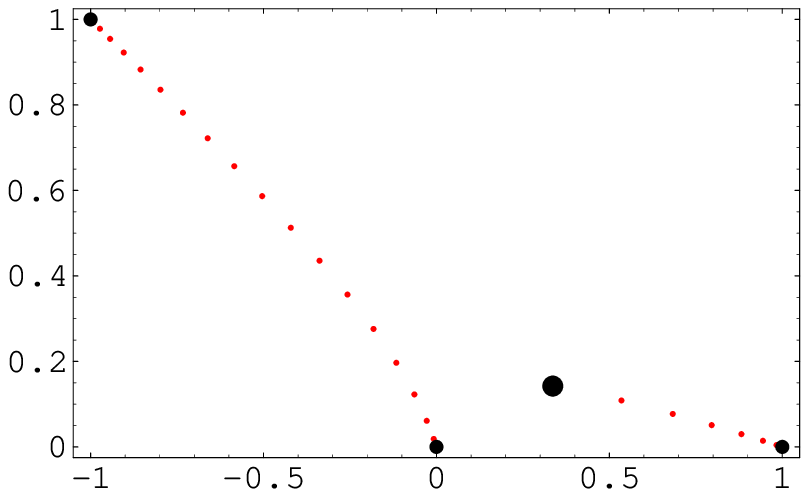}}
\hskip0.5cm\hbox{\epsfysize=1.6cm\epsfbox{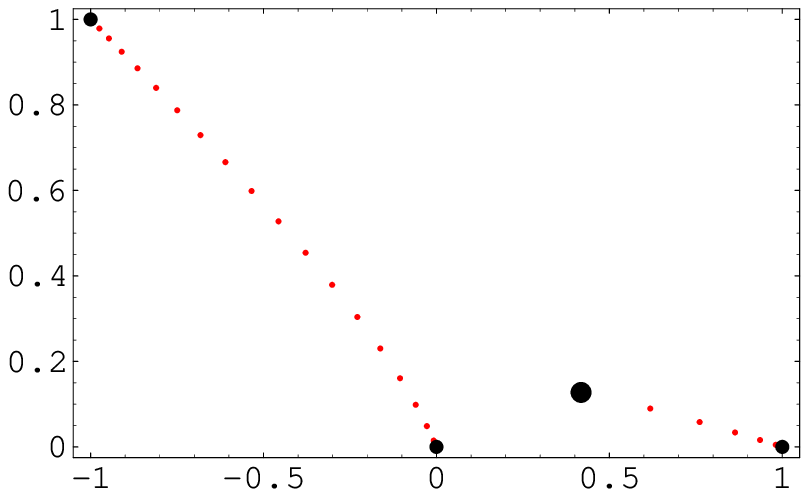}}
}

\centerline{\hbox{\epsfysize=1.6cm\epsfbox{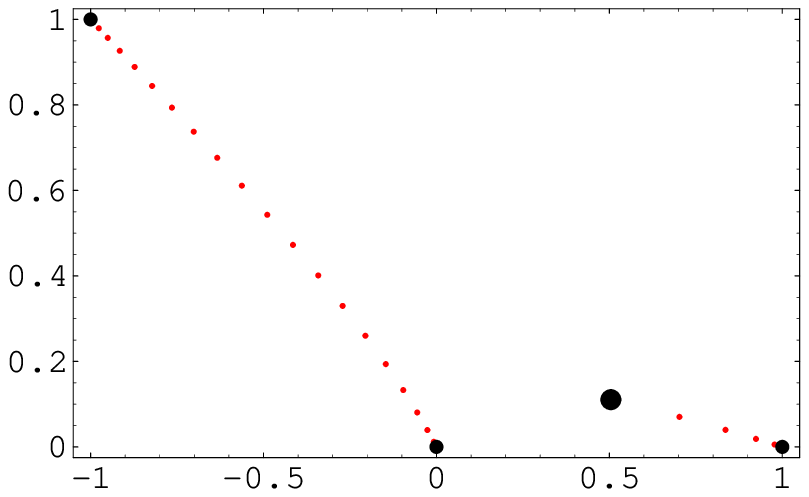}}
\hskip0.5cm\hbox{\epsfysize=1.6cm\epsfbox{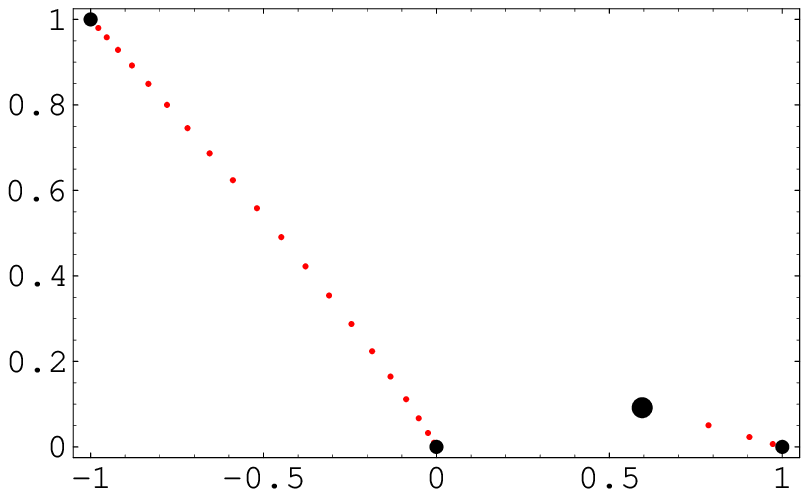}}
\hskip0.5cm\hbox{\epsfysize=1.6cm\epsfbox{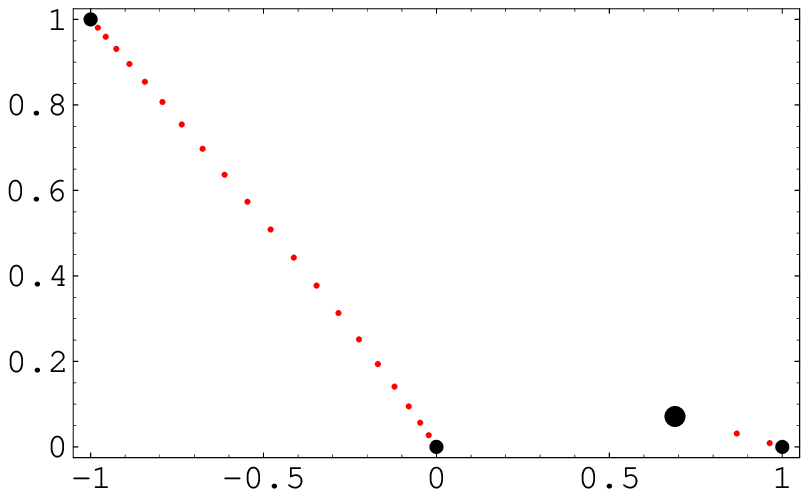}}
\hskip0.5cm\hbox{\epsfysize=1.6cm\epsfbox{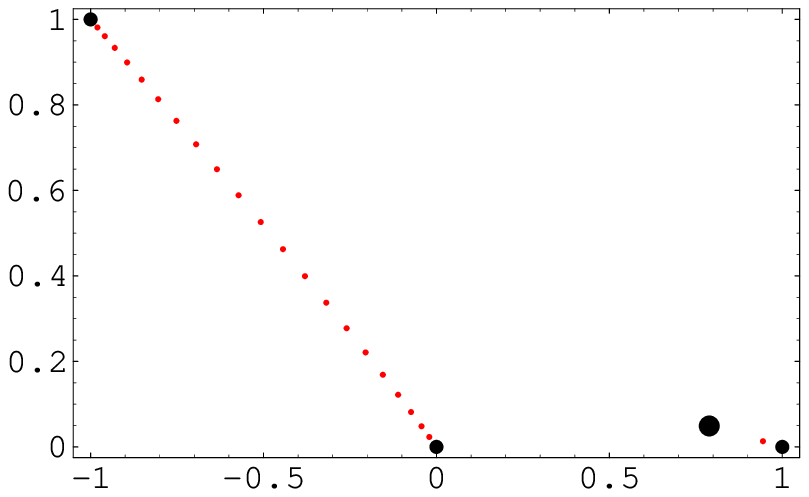}}
\hskip0.5cm\hbox{\epsfysize=1.6cm\epsfbox{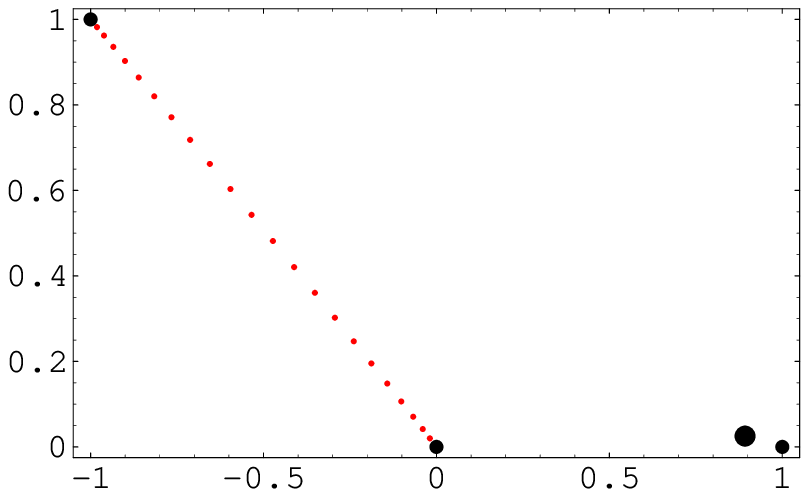}}
}

\vskip 1cm
{Figure 3. Zeros of $25$ different Stieltjes polynomials of degree $24$
for the equation $Q(z)S^{\prime\prime}(z)+V(z)S(z)=0$ with
$Q(z)=z(z-1)(z-1+I)$. }
\label{fig2}
\end{figure}

\begin{remark} In fact, we will prove our  results in the reverse order starting with Theorems~\ref{th:charac}, ~\ref{pr:positive},   then \ref{th:higRull} and, finally,   settling Theorem~\ref{takemura}. 
This order is necessary since the convergence and unicity statements in Theorem~\ref{takemura}  require some technique and facts  from the former theorems.
\end{remark}

\medskip
\noindent 
{Explanation to Figure~3.} 
The smaller dots on each of the $25$ pictures above are the $24$ zeros of the corresponding Stieltjes polynomial $S(z)$; the $3$ average size dots are the zeros of $Q(z)$ and the single large dot is the (only) zero of the corresponding $V(z)$.

\begin{ack}
 We are obliged to Professor A.~Kuijlaars for clarification of his joint paper \cite{KvA} and to Professor A.~Mart\'inez-Finkelshtein for the interest in our work as well and many discussions of this topic. In the late stage of working on this project we discovered that he and Professor E.~Rakhmanov were studying   very similar questions, see \cite{MFR1}, \cite{MFR2} and our results have serious  intersection although were obtained completely independently and by using rather different technical tools. In particular, the above Theorem~\ref{takemura}Ê is contained in  \cite {MFR1}.  Sincere thanks go  to Thomas Holst and especially to Yuliy Baryshnikov for their help with quadratic differentials. 
 The second author wants to acknowledge the hospitality of the department of mathematics, Stockholm University during his visit to Stockholm  in Fall 2007 when this project was initiated. He was supported by the Grant-in-Aid for Young Scientists (B) (No. 19740089) from the Japan Society for the Promotion of Science. Research of the third author  was supported by the Czech Ministry of Education, Youth and Sports within the project LC06002.
\end{ack}

\medskip 
\section{Proving Theorems~\ref{th:charac} and ~\ref{pr:positive}}\label{sec:2}

It is well-known that any closed trajectory $\gamma$ of a Strebel differential $\Psi$ is contained in
the maximal connected domain $D_{\gamma}$, completely filled with
closed trajectories, such that any closed trajectory $\gamma'$ is
contained in $D_{\gamma}$ if, and only if, it is homotopic to
$\gamma$ on the corresponding Riemann surface, see \cite{Str}. This
implies, in particular, that any non-closed trajectory of a Strebel
differential is a part of the boundary of one of these connected
components consisting of closed trajectories. Such a boundary
consists of singular points and singular trajectories since
otherwise the boundary will be a closed trajectory  itself, which
clearly contradicts  to the maximality of the corresponding domain. 
We now  prove the compactness result  mentioned in the Introduction in 
 case of rational differentials, comp.  Theorem~20.1 in \cite{Str}. 

\begin{lemma}
\label{lemma2} Let  $U_1(z)$ and $U_2(z)$  be monic polynomials with 
$\deg U_2- \deg U_1=2$. The rational differential $\Psi=-\frac{U_1(z)}{U_2(z)}dz^2$ is
Strebel if and only if  the union $K_\Psi$ of all its singular trajectories and
singular points is compact.
\end{lemma}

\begin{proof} We need to show that a rational Strebel differential $\Psi$ as above always has a compact set of closed trajectories and, conversely, that any rational differential as above belonging to the class 
$\M$ having a compact set of closed trajectories is Strebel. 

To prove the first implication   
 notice that due to  the assumption on the degrees of $U_1(z)$
and $U_2(z)$, the point $z=\infty$ will be a pole of order 2.
Since $U_1(z)$ and $U_2(z)$ are monic, the leading term of
$-\frac{U_1(z)}{U_2(z)}$ at infinity is negative, and hence there
is a neighborhood $D$ of $\infty$ such that
$D\setminus\{\infty\}$ is filled with closed trajectories. Hence the
union of all singular trajectories and singular points is contained
in a compact set $K\subset \bC$. Take an infinite sequence of points $\{z_n\}$,
all lying on some (not necessarily the same) singular trajectory. By
compactness of $K$, there is a converging subsequence
$z_{n_i}\rightarrow{}z^*$. The point $z^*$ can  not lie on a closed
trajectory, since it would then lie in an open domain free from
singular trajectories. Hence it must lie on the boundary of a domain
filled with closed trajectories. Thus $z^*$ is either a singular
point or lies on a singular trajectory. 

The proof in converse direction follows closely  that of Theorem~20.1 in \cite{Str} and is omitted here. 
\end{proof}

The next proposition is  important for our main construction in Theorem~\ref{th:charac}.

\begin{lemma}
\label{lemma3} For any rational Strebel differential $\Psi=-\frac{U_1(z)}{U_2(z)}dz^2$ 
 with monic  $U_1(z)$ and $U_2(z)$ satisfying $\deg U_2 -\deg U_1 =2$  there exists a compact union
$\frak U$ of its singular trajectories and singular
points such that:\\
1) all its singular points (except $\infty$) lie in $\frak U$;\\
2) $\mathbb{C}\setminus{}\frak U$ is connected.
\end{lemma}

\begin{proof}
Let $\frak U_0=K_\Psi$ be the union of all singular trajectories and all 
singular points  of $-\frac{U_1(z)}{U_2(z)}dz^2$ except $\infty$. By lemma 2, the
complement  $D_0:=\mathbb{C}\setminus{}\frak U_0$ is an open set with the
unique unbounded component $A_\infty$ (containing $\infty$) and finitely many bounded components
$B_1,...,B_k$. We will remove some of the singular trajectories
from $\frak U_0$ to obtain the required $\frak U$. Let \textit{\textbf{B}} be
the collection $\{B_{k_i}\}$ of all bounded components such that
$\partial{}A_\infty\cap\partial{}B_{k_i}\neq\emptyset$. Let $z_0$ be a
point in $\partial{}A_\infty\cap\partial{}B_{k_i}$ for some
$B_{k_i}\in\textit{\textbf{B}}$. If $z_0$ is a regular point, then
the trajectory $\gamma_{z_0}$ is  singular  and
$\gamma_{z_0}\subset\partial{}A_\infty\cap\partial{}B_{k_i}$. Put
$\frak U_1:=\frak U_0\setminus\gamma_{z_0}$. Then the set
$D_1:=\mathbb{C}\setminus{}\frak U_1$ contains one component less than $D_0$
and its unbounded component is $A_\infty\cup\gamma_{z_0}\cup{}B_{k_i}$ for
some $B_{k_i}\in\textit{\textbf{B}}$. Assume now that $z_0$ is a
singular point. It cannot be a pole of order two, because then it
would be isolated and constitute a boundary component of a unique
component of $D_0$. It cannot be a pole of order one either, since
then there would be a unique singular trajectory entering $z_0$. It
must thus be a zero of some order. Local analysis of zeros of
quadratic differentials (see Ch. 2,  \cite{Str}) shows that there is a finite number of
singular trajectories entering $z_0$, which split a small
neighborhood $N_{z_0}$ of $z_0$ into a finite number of open
sectors. The intersection $A_\infty\cap{}N_{z_0}$ is obviously non-empty
and consists of some of the open sectors, but not all of them (since
$z_0\in\partial{}A_\infty\cap\partial{B_{k_i}}$ for some
$B_{k_i}\in\textit{\textbf{B}}$). Thus, there exists a singular
trajectory entering $z_0$ which lies on the common boundary of $A_\infty$
and $B_{k_j}$ for some $B_{k_j}\in\textit{\textbf{B}}$. We can thus
remove one of these singular trajectories, call it $\gamma$, from
$U_0$ to obtain a smaller set $\frak U_1:=\frak U_0\setminus\gamma$ such that
$D_1:=\mathbb{C}\setminus{}\frak U_1$ contains one component less than $D_0$
and its unbounded component is $A_\infty\cup\gamma\cup{}B_{k_j}$ for some
$B_{k_j}\in\textit{\textbf{B}}$. Continuing this process (which
must end after a finite number of steps) we obtain the final set $\frak U$
satisfying conditions 1) and 2). (Notice that, in general, $\frak U$ is not unique.) 
\end{proof}


To move further  we  need some information about compactly supported real
measures and their Cauchy transform.

\begin{lemma} [comp. Th. 1.2, Ch. II, \cite{Ga}]
\label{lemma4} Suppose $f\in{}L_{loc}^1(\mathbb{C})$ and that
$f(z)\rightarrow0$ as $z\rightarrow\infty$ and let $\mu$ be a
compactly supported measure in $\mathbb{C}$ such that
$\frac{\partial{}f}{\partial\bar{z}}=-\pi\mu$ in the sense of
distributions. Then $f(z)=C_{\mu}(z)$ almost everywhere, where
$C_{\mu}(z)=\int_{\mathbb{C}}\frac{d\mu(\xi)}{z-\xi}$ is the Cauchy
transform of $\mu$.
\end{lemma}

\begin{proof}
It is clear that $C_{\mu}$ is locally integrable, analytic off the
closure of the  support of $\mu$ and vanishes at infinity. Considering
$h=f-C_{\mu}$ and assuming that $h$ is a locally integrable function
vanishing at infinity and satisfying
$\frac{\partial{}h}{\partial\bar{z}}=0$ in the sense of
distributions. We must show that $h=0$ almost everywhere. Let
$\phi_r\in{}C_0^{\infty}(\mathbb{C})$ be an approximate to the
identity, i.e. $\phi_r\geq0$, $\int_{\mathbb{C}}\phi_rdxdy=1$ and
$supp({\phi_r})\subset\{|z|<r\}$, and consider the convolution
\[h_r(z)=h*\phi_r=\int_{\mathbb{C}}h(z-w)\phi_r(w)dxdy, \ w=x+iy.\]
It is well known that $h_r\in{}C^{\infty}$ and that
$\lim\limits_{r\rightarrow0}h_r\rightarrow{}h$ in $L^1(K)$ for any
compact set $K$. Moreover
\[\frac{\partial{}h_r}{\partial\bar{z}}=\frac{\partial{}h}{\partial\bar{z}}*\phi_r=0.\]
This shows that $h_r$ is an entire function which vanishes at
infinity, implying that $h_r\equiv0$. Hence $h=0$ a.e.
\end{proof}
Lemmas \ref{lemma3} and \ref{lemma4} show how to construct the 
unique measure whose support lies in some subset $\frak U$ of singular points
and singular trajectories of $-\frac{U_1(z)}{U_2(z)}dz^2$  by
specifying a  branch of $\sqrt{\frac{U_1(z)}{U_2(z)}}$ in
$\mathbb{C}\setminus{}\frak U$. 

\begin{proof} [Proof of Theorem~\ref{th:charac}] First we prove that  given a Strebel differential 
$\Psi=-\frac{U_1(z)}{U_2(z)}dz^2$ one can construct a real measure supported on  $K_\Psi$ with the required properties. To do this choose the union $\frak U$ of singular
trajectories of $\Psi$ as in Lemma 3. It has the property that we can choose a single-valued branch of $\sqrt{\frac{U_1(z)}{U_2(z)}}$ in $\mathbb{C}\setminus{\frak U}$
which behaves as $\frac{1}{z}$ at infinity. Define
$$\mu:=\frac{\partial\sqrt{\frac{U_1(z)}{U_2(z)}}}{\partial\bar{z}}$$
 in the sense of distributions. The distribution $\mu$ is evidently compactly supported on $\frak U$.
By lemma 4 we get that $C_{\mu}(z)$ satisfies \eqref{eq:imp} a.e in $\bC$. 

It remains to show
that $\mu$ is a real measure. Take a point $z_0$ in the support of
$\mu$ which is a regular point of $-\frac{U_1(z)}{U_2(z)}dz^2$ and take a small
neighborhood $N_{z_0}$ of $z_0$ which does not contain  roots of
$U_1(z)$ and $U_2(z)$. In this neighborhood we can choose a single branch
$B(z)$ of $\sqrt{\frac{U_1(z)}{U_2(z)}}$. Notice that $N_{z_0}$ is divided
into two sets by the support of $\mu$ since it by
construction consists of singular trajectories of $\mu$. Denote  these
 sets  by $Y$ and $Y'$ resp. Choosing $Y$ and $Y'$
appropriately we can represent $C_{\mu}$ as $\chi_{Y}B-\chi_{Y'}B$
in $N_{z_0}$ up to the support of $\mu$, where $\chi_{X}$ denotes
the characteristic function of the set $X$. By a theorem 2.15 in
\cite{Her}, ch.2  we have
\[<\mu,\phi>=<\frac{\partial{}C_{\mu}}{\partial\bar{z}},\phi>=<\frac{\partial(\chi_{Y}B-\chi_{Y'}B)}{\partial\bar{z}}>=i\int_{\partial{}Y}B(z)\phi{}dz,
\]
for any test function $\phi$ with compact support in $N_{z_0}$.
Notice that the last equality holds because $\phi$ is identically
zero in a neighborhood of $\partial{}N_{z_0}$, so it is only on the
common boundary of $Y$ and $Y'$ that we get a contribution to the
integral given in the last equality. This common boundary  is the singular trajectory $\gamma_{z_0}$ intersected with the neighborhood $N_{z_0}$. The integral
\[i\int_{\partial{}Y}B(z)\phi{}dz\]
is real since the change of coordinate $w=\int_{z_0}^{z}iB(\xi)d\xi$
transforms the integral to the integral of $\phi$ over a part of the
real line. This shows that $\mu$ is locally a real measure, which 
proves one implication of the theorem.

To prove  that a compactly supported real 
measure $\mu$  whose Cauchy transform satisfies \eqref{eq:imp}Ê everywhere except for a set of measure zero produces the Strebel differential $\Psi=-\frac{U_1(z)}{U_2(z)}dz^2$ consider 
 its logarithmic potential $u_{\mu}(z)$, i.e. 
\[u_{\mu}(z):=\int_{\mathbb{C}}\log|z-\xi|d\mu(\xi).\]
The function $u_{\mu}(z)$ is harmonic outside the
support of $\mu$, and  subharmonic in the whole $\mathbb{C}$. The following important relation: 
\[\frac{\partial{u_{\mu}(z)}}{\partial{z}}=\frac{1}{2}C_{\mu}(z)\]
connects  $u_{\mu}$ and $C_{\mu}$. 
It implies that the set of level curves of $u_{\mu}(z)$ coincides 
with the set of horizontal trajectories of the quadratic differential
$-C_{\mu}(z)^2dz^2$. Indeed, the gradient of $u_{\mu}(z)$ is given
by the vector  field with coordinates $\left(\frac{\partial{u_{\mu}}}{\partial{x}},\frac{\partial{u_{\mu}}}{\partial{y}}\right)$ in $\bC$. Such a vector in $\bC$ coincides with the complex number
$2\frac{\partial{u_{\mu}}}{\partial\bar{z}}$. Hence, the gradient of
$u_{\mu}$ at $z$ equals to $\bar{C}_{\mu}(z)$ (i.e. the complex conjugate of
$C_{\mu}(z)$). But this is the same as saying that the vector
$i{}\bar{C}_{\mu}(z)$ is orthogonal to (the tangent line to) the level
curve of $u_{\mu}(z)$ at every point $z$ outside the support of
$\mu$. Finally, notice that at each point $z$ one has 
$$-C_{\mu}^2(z)(i\bar{C}^2_{\mu}(z))>0,$$
 which exactly means
 that the horizontal trajectories of $-C_{\mu}(z)^2dz^2=-\frac{U_1(z)}{U_2(z)}dz^2$ are
 the level curves of $u_{\mu}(z)$ outside the support of
$\mu$.  Notice that $u_{\mu}(z)$ behaves as $\log {|z|}$ near $\infty$ and  is continuous except for
 possible second order poles where it has logarithmic singularities with a negative leading coefficient. This guarantees that almost all its level curves are closed and smooth  implying that  $-\frac{U_1(z)}{U_2(z)}dz^2$ is Strebel. 

Let us  settle part (2) of  Theorem~\ref{th:charac}.  Notice that if a real measure whose Cauchy transform satisfies  \eqref{eq:imp} a.e. is supported on the compact set $K_\Psi$ (which consists of finitely many singular trajectories and singular points)  all one needs to determine it uniquely is just  to prescribe which of the two branches of $\sqrt{\frac{U_1(z)}{U_2(z)}}$ one should choose as the Cauchy transform of this measure  in each bounded connected component of the complement $\bC\setminus K_\Psi$. (The choice of a branch in the infinite component is already prescribed by the requirement that it should behave as $\frac{1}{z}$ near infinity.)  Any such choice of branches in open domains leads to a real measure, see proof of part (1) above. 

In order to explain the details shown on  Fig. 2 notice that singular trajectories are level sets of the logarithmic potential of a real measure under consideration and, therefore, the gradient of this potential is perpendicular to any such trajectory. (More exactly, the logarithmic potential $u(z)$ is continuous at a generic point of any singular trajectory and its gradient has at least one-sided limits when $z$ tends to such a generic point. These one-sided limits are necessarily perpendicular to the trajectory and they either coincide or are the opposite of each other.)  So the choice of a branch of $\sqrt{\frac{U_1(z)}{U_2(z)}}$  in some connected component of $\bC\setminus K_\Psi$ can be uniquely determined and restored from the choice of direction of the gradient of the logarithmic potential near some singular trajectory belonging to the boundary of this component. One can easily see that if gradients on both sides of a certain singular trajectory belonging to $K_\Psi$ have the same direction then one chooses the same branch of $\bC\setminus K_\Psi$ on both sides and its $\partial \bar z$-derivative vanishes  on this singular trajectory leaving it  outside the support of the corresponding measure. If these gradients have opposite directions then in case when they are both directed away from the trajectory the measure on this trajectory will be positive and in case when they  are both directed towards the trajectory the measure on it will be negative. These observations explain how to obtained all supports and signs of measures appearing as the result of $2^{d-1}$ different choices of branches in $d-1$ bounded connected components.  \end{proof}

Using the latter observations we now prove Theorem~\ref{pr:positive} and   formulate a criterion of existence of a positive measure in terms of certain properties of $K_\Psi$. We need some new definitions. 

\begin{df} Associate to each connected component (alias {\it domain}) of $\bC\setminus K_\Psi$  its {\it depth}  as follows. 
Set the depth of the infinite component $A_\infty$ to $0$. Each component neighboring to the infinite one gets depth $1$. Each new component neighboring to a component with depth $i$ gets the depth $i+1$ etc. 
\end{df} 

\begin{df}  We call a singular trajectory in $K_\Psi$ {\em dividing} if it belongs to the closure of two different domains  and {\em non-dividing} otherwise. 
\end{df}

Notice that since any  domain is homeomorphic to an open annulus then its boundary consists of two connected components which are either homotopy equivalent to $S^1$ or to a point. We will call  the one which separates the considered domain from the domain of the smaller depth  the {\em outer} boundary and the other one the {\em inner} boundary.  

\begin{df} A non-dividingÊ singular trajectory is called {\em preventing} (resp. {\em non-preventing}) if it is a part of the outer (resp. inner) boundary of its open component of $\bC\setminus K_\Psi$. 
\end{df}

\begin{center}
\begin{picture}(440,170)(0,0)
\put(130,140){\line(1,0){60}}
\put(130,140){\circle*{3}}
\put(190,140){\circle*{3}}

\qbezier(160,165)(210,195)(190,140)
\qbezier(160,165)(110,195)(130,140)
\qbezier(160,105)(210,110)(190,140)
\qbezier(160,105)(110,110)(130,140)
\put(160,165){\circle*{3}}
\put(160,165){\line(0,-1){12}}
\put(160,153){\circle*{3}}

\put(180,150){\circle*{3}}
\put(140,150){\circle*{3}}
\qbezier(140,150)(160,142)(180,150)

\put(140,125){\circle*{3}}
\put(180,125){\circle*{3}}
\qbezier(140,125)(160,115)(180,125)

\put (210, 152){\vector(-3,-1){30}}
\put(220,152){Dividing trajectory separating}
\put(220,142) {domains of the same depth = 1}

\put (98, 160){\vector(1,0){60}}
\put(48,160) {Preventing}
\put(48,150)  {trajectory}

\put(50,60){Figure 4. Some types of singular trajectories}
\end{picture}
\end{center}

We can finally formulate a criterion of the existence of a positive measure.

\begin{prop}\label{pr:crit} A Strebel differential $\Psi=-\frac{U_1(z)}{U_2(z)}dz^2$ admits a positive measure satisfying \eqref{eq:imp} if and only if
\begin{itemize}
\item no dividing singular trajectory separates two open domains of the same depth;

\item each non-dividing singular trajectory is non-preventing.

\end{itemize}Ê \end{prop}

\begin{proof} To get a positive measure $\mu$ one needs that the gradient of its logarithmic potential 
$u_\mu(z)$ is directed outwards on the outer boundary of each domain.  Indeed, if it is directed inward on the outer boundary then it is also directed inward on the inner boundary which leads to negativity of the density (probably on the boundary of some domain with a higher  depth). That  is  if this  inner boundary is not the outer boundary of a higher level then you get a negative density. In case it is the outer boundary of a domain of the next level you can  escape negativity on this inner boundary by directing the gradient inside on the next level. But that means that the situation persists on the next level. 
SInce the number of domains is finite one gets a negative density anyway.  Let us show that the above two conditions are necessary. Indeed, if you have a dividing trajectory separating two domains of the same depth you get the gradient directed towards it from both sides which gives  a negative density. Analogously, if you have a preventing trajectory then it gets negative density for the same reason. 

Assume now that none of this happens. Then our choice of branches of $\C_\mu$ (outward in each domain) determines the support of $\mu$ uniquely. Namely, all dividing singular trajectories  disappear since they connect different levels. What remains  are some non-dividing trajectories connected to the inner part on the deepest level but they all have positive density.
\end{proof}

\begin{proof}[Proof of Theorem~\ref{pr:positive}] It follows immediately from the proof of Proposition~\ref{pr:crit}. The only detail which needs to be clarified is why the support of a positive measure whose Cauchy transform satisfies 
\eqref{eq:imp} a.e. must necessarily belong to $K_\Psi$. Indeed, by general results of \cite {BR}Ê the support is locally made of horizontal trajectories of $\Psi$. Moreover, it must be constructed of the whole trajectories since its vertex is necessarily a singular point of $\Psi$. Assume that this support contains some  closed trajectory which is not in $K_\Psi$. Among such closed trajectories there should be one which is the inner boundary of the infinite domain. Then inside it we should choose the branch of $\sqrt{R(z)}$ different from that in the infinite domain. Thus the gradient of the logarithmic potential inside this trajectory should be directed inward.  But then the same argument  as  in the proof of Proposition~\ref{pr:crit}Ê shows that there is a part of the support of this measure lying inside the considered trajectory where the measure is forced to be negative.
\end{proof}

           \section {Proving Theorem~\ref{th:higRull}}
	 \label{sec:proofs}   

Our scheme follows roughly the
scheme suggested in \cite {BR}.  We  need to prove that under the assumptions of Theorem~\ref{th:higRull} the sequence $\{\nu_{n,i_n}\}$ of root-counting measures of the sequence of Stieltjes polynomials $\{S_{n,i_n}(z)\}$ converges weakly to a probability 
	   measure $\nu_{\tilde V}$ whose Cauchy transform $\C_{\widetilde V}(z)$ satisfies almost everywhere in $\bC$ the equation 
	   \begin{equation}\label{eq:again}
	   \C_{\widetilde V}^2(z)=\frac {\widetilde V(z)}{Q(z)}.
	   \end{equation}
	Since such a measure is positive it is unique by Theorem~\ref{th:charac}  which implies  Corollary~\ref{cor:main}.  
	   
	   To simplify the notation  we denote by  $\{\bar S_n(z)\}$  the chosen sequence  $\{S_{n,i_n}(z)\}$ of  Stieltjes polynomials whose sequence of normalized Van Vleck polynomials $\{\tilde V_{n,i_n}(z)\}$ converge to $\widetilde V(z)$ and let $\{\bar \nu_n\}$ denote the sequence of its root-counting measures. Also let $\bar \nu_n^{(i)}$ be the root measure of the $i$th derivative $\bar S^{(i)}_n(z)$. Assume now that $NN$ is a subsequence of natural numbers such that 
	   $$\bar \nu^{(j)}=\lim_{n\to \infty,n\in NN}\bar \nu_n^{(j)}$$
	   exists for $j=0,1,2$. 
	   The next lemma shows that   the Cauchy transform of $\bar \nu=\bar \nu^{(0)}$ satisfies the required algebraic equation.
	   
	   \begin{lemma}\label{lm:CauEq} The measures $\bar \nu^0, \bar \nu^1, \bar \nu^2$ are all equal and the Cauchy transform $\C_{\tilde V}(z)$ of their common limit satisfies the equation \eqref{eq:again}  for almost every $z$. 
	   	   \end{lemma} 
	   
	   \begin{proof}ÊWe have 
	   $$\frac {\bar S_n^{(j+1)}(z)}{(n-j)\bar S_n^{(j)}(z)}\to \int\frac {d\bar\nu^{(j)}(\zeta)}{z-\zeta}$$ 
	   with convergence in $L^{1}_{loc}$, and by passing to a subsequence again we can assume that we have pointwise convergence almost everywhere. From the relation 
	   $$Q(z) \bar S_n^{\prime\prime}(z)+P(z)\bar S^\prime_n(z)+V_n(z)\bar S_n(z)=0$$
	    it follows that 
	   $$\frac{Q(z)\bar S^{\prime\prime}_n(z)}{n(n-1)\bar S_n(z)}+\frac{V_n(z)}{n(n-1)}= -
	   \frac{P(z)\bar S_n^\prime(z)}{n(n-1)\bar S_n(z)}.$$
	      
	      One can immediately check that $-\frac{V_n(z)}{n(n-1)}\to \widetilde V(z)$, while the expression in the right-hand side converges pointwise to $0$ almost everywhere in $\bC$ due to presence of the factors $n(n-1)$ in the denominator. Thus,  for almost all $z\in\bC$ one has 
	      $$ \frac{\bar S^{\prime\prime}_n(z)}{n(n-1)\bar S_n(z)}\to \frac {\widetilde V(z)}{Q(z)}$$ 
	      when $n\to\infty $ and $n\in NN$. If $u^{(j)}$ denotes the logarithmic potential of $\bar \nu^{(j)}$, then 
	      one has 
	      $$u^{(2)}-u^{(0)}=\lim_{n\to\infty} \frac 1 n\log \left\vert \frac {\bar S^{\prime\prime}_n(z)}{n(n-1)\bar S_n(z)}\right\vert=\lim_{n\to\infty}\frac 1 n 
	      \left(\log\vert \widetilde V(z)\vert -\log \vert Q(z)\vert\right)=0.$$
	      On the other hand we have that $u^{(0)}\ge u^{(1)}\ge u^{(2)}$, see Lemma~\ref{lm:logpot} below. Hence all the potentials $u^{(j)}$ are equal, and all $\nu_j=\Delta u^{(j)}/2\pi$ are equal as well. Finally we get 
	      $$C_{\tilde V}^{2}(z)=\lim_{n\to\infty}\frac{\bar S_n^{\prime}(z)}{n\bar S_n(z)}\cdot \frac{\bar S_n^{\prime\prime}(z)}{(n-1)\bar S_n^\prime(z)}=\lim_{n\to\infty}\frac {\bar S_n^{\prime\prime}(z)}{n(n-1)\bar S_n(z)}=\frac {\widetilde V(z)}{Q(z)}$$ 
	      for almost all $z$. 
	      	      
\begin{lemma} [\rm{see Lemma 8 of \cite{BR}}] \label{lm:logpot} Let $\{p_m\}$ be a sequence of polynomials, such that $n_m:=\deg p_m\to \infty$ and there exists a compact set $K$ containing the zeros of all $p_m$. Finally, let $\mu_m$ and $\mu'_m$ be the root-counting measures of $p_m$ and $p_m'$ resp. If $\mu_m\to\mu$ and $\mu_m'\to\mu'$ with compact support and $u$ and $u'$ are the logarithmic potentials of $\mu$ and $\mu'$, then $u'\le u$ in the whole $\bC$. Moreover,  $u=u'$ in the unbounded component $A_\infty$ of $\bC\setminus \text{supp } \mu$. 
\end{lemma} 

\begin{proof} Assume wlog that $p_m$ are monic. Let $K$ be a compact  set containing the zeros of all $p_m$.  We have 
$$u(z)=\lim_{m\to\infty}\frac{1}{n_m}\log \vert p_m(z)\vert$$ 
and 
$$u'(z)=\lim_{m\to\infty}\frac{1}{n_m-1}\log\left\vert\frac{p'_m(z)}{n_m}\right\vert=\lim_{m\to\infty}\frac{1}{n_m}\log\left\vert\frac{p'_m(z)}{n_m}\right\vert$$
with convergence in $L_{loc}^1$.  Hence 
$$u'(z)-u(z)=\lim_{m\to\infty}\frac{1}{n_m}\log\left\vert\frac{p'_m(z)}{n_m p_m(z)}\right\vert=\lim_{m\to\infty}\frac{1}{n_m}\log\left\vert \int \frac{d\mu_m(\zeta)}{z-\zeta}\right\vert.$$Ê
Now, if $\phi$ is a positive test function it follows that
$$\int\phi(z)(u'(z)-u(z))d\la(z)=\lim_{m\to\infty}\int\phi(z)\log\left\vert\int \frac{d\mu_m(\zeta)}{z-\zeta}\right\vert d\la(z)\le $$
$$\le \lim_{m\to\infty}\int\phi(z)\int \frac{d\mu_m(\zeta)}{\vert z-\zeta \vert} d\la(z)\le 
\lim_{m\to\infty}\iint\frac{\phi(z)d\la(z)}{\vert z-\zeta\vert}  d\mu_m(\zeta)$$
where $\la$ denotes Lebesgue measure in the complex plane. Since $\frac{1}{\vert z\vert}$ is locally integrable, the function $\int \phi(z)\vert z-\zeta\vert^{-1}d\la(z)$ is continuous, and hence bounded by a constant $M$ for all $z$ in $K$. Since $\text{supp }\mu_m\in K$, the last expression in the above inequality is bounded by $M/n_m$, hence the limit when $m\to\infty$ equals to $0$. This proves $u'\le u$. 

In the complement to $\text{supp }\mu$, $u$ is harmonic and $u'$ is subharmonic, hence $u'-u$ is a negative subharmonic function. Moreover, in the complement of $K$, $p_n'/(n_mp_m)$ converges uniformly on compact sets to the Cauchy transform $\C_\mu(z)$ of $\mu$. Since $\C_\mu(z)$ is a non-constant  holomorphic function in the unbounded component $A_\infty$ of $\bC\setminus K$, then by the above $u'-u=0$ there. By the maximum principle for subharmonic functions it follows that $u'-u=0$ holds in the unbounded component of $\bC\setminus \text{supp }\mu$ as well. 
\end{proof}Ê

	      To accomplish the proof of Theorem~\ref{th:higRull} we need to show that we have the convergence for the whole sequence and not just for some subsequence.  Assume now that the sequence $\bar \nu_n$ does not converge to $\bar \nu$. Then we can find a subsequence  
	       $NN'$ such that $\bar \nu_n$ stay away from some fixed neighborhood of $\bar \nu$ in the weak topology, for all $n\in NN'$. Again by compactness, we can find a subsequence $NN^{*}$ of $NN'$ such that all the limits for root measures for derivatives 
	       exist for $j=0,...,k$. But then $\bar \nu^{(0)}$ must coincide with $\bar \nu$ by the uniqueness and the latter lemma. We get a contradiction to the assumption that $\nu_n$ stay away from $\bar \nu$ for all $n\in NN'$ and hence all $n\in NN$.  	   
	   \end{proof}

\section{On Strebel differentials of the form $\Psi=\frac{(b-z)dz^2}{(z-a_1)(z-a_2)(z-a_3)}$}

The main result of this section is as follows. 

\begin{theorem}\label{th:struct} For a given Strebel differential $\Psi$ as in the title  the union of its singular trajectories starting at $a_1,a_2,a_3$  is contained in  the convex hull $\Delta_Q$ of these roots 
if and only if $b\in \Ga_Q$ where  $Q(z)= (z-a_1)(z-a_2)(z-a_3)$. 

\end{theorem}Ê

\begin{rem}
For the definition of $\Ga_Q$ see Introduction. The proof below was suggested by the second author. A completely different proof was later found by Y.~Baryshnikov based on his interpretation of interval exchange transformations in our situation.  
\end{rem}

Assume that the points $a_1, a_2, a_3$ are not collinear in the complex plane. 
Let $i,j,k$ be a permutation of $1,2,3$.
Recall that we  defined the curve $\ga_i,\; i=1,2,3$ by  the condition: 
\begin{align}
\Im f_{j,k}(b)=0 \text{ \quad  where  \quad} &f_{jk}(b)= \int _{a_j}^{a_k}  \sqrt{\frac{b-t}{(t-a_1)(t-a_2)(t-a_3)}} dt.
\label{eq:ba123ReaN}
\end{align}

\begin{lemma}\label{lm:trans} Each $\ga_i$ is smooth and they can only intersect transversally.
\end{lemma}

\begin{proof}
Indeed, we have 
\begin{align}
& f_{jk}^\prime(b) = \int _{a_j}^{a_k}  \sqrt{\frac{-1}{(t-a_1)(t-a_2)(t-a_3)(t-b)}} dt.
\label{eq:a123bRea}
\end{align}
Since the right-hand side in (\ref{eq:a123bRea}) is a complete elliptic integral, it represents a period of an elliptic curve which implies that the right-hand side is nonvanishing which in its turn implies  the smoothness of $\ga_i$. To show that $\ga_i$ and $\ga_j$ for $i\neq j$ can only intersect transversally notice the following. If they are tangent at some point  $b^*$ then $f^\prime_{jk} (b^*)/f^\prime_{ik} (b^*) \in \Rea$ but this can  never happen since the ratio of periods of an elliptic curve cannot be  real.
\end{proof}

\begin{lemma} The following 3 relations hold: 
\begin{align}
& \int _{a_j}^{a_k}  \sqrt{\frac{a_i-t}{(t-a_1)(t-a_2)(t-a_3)}} dt = \int _{a_j}^{a_k}  \sqrt{\frac{1}{(a_j-t)(t-a_k)}} dt = \pi ,\\
& \int _{a_j}^{a_k}  \sqrt{\frac{b-t}{(t-a_1)(t-a_2)(t-a_3)}} dt + \int _{b}^{a_i} \sqrt{\frac{b-t}{(t-a_1)(t-a_2)(t-a_3)}} dt = \pi ,
\end{align}

\begin{align}
& \left( \int _{a_j}^{a_k}  +\int _{a_k}^{a_i} +\int _{a_i}^{a_j}  \right) \sqrt{\frac{b-t}{(t-a_1)(t-a_2)(t-a_3)}} dt = 2\pi . \label{eq:threeint}
\end{align}

\end{lemma}

\begin{proof}
To prove the first relation  make an affine change of variable  $\tilde{z}= c z +d$ where $c\neq 0$.
Set $\tilde{a} _i=ca_i +d$ and $\tilde{b}=cb +d $.
If $a_1, a_2, a_3, b $ correspond  to $\tilde{a}_1, \tilde{a}_2, \tilde{a}_3, \tilde{b} $ resp. then we have
\begin{align}
& \int _{z'}^z  \sqrt{\frac{b-t}{(t-a_1)(t-a_2)(t-a_3)}} dt = \int _{\tilde{z}'} ^{\tilde{z}}  \sqrt{\frac{\tilde{b}-\tilde{t}}{(\tilde{t}-\tilde{a}_1)(\tilde{t}-\tilde{a}_2)(\tilde{t}-\tilde{a}_3)}} d\tilde{t}, \label{eq:invaffinetr}
\end{align}
where $\tilde{t}= ct +d$ and $\tilde{z}'= cz' +d$.
Hence we can always place two points $a_j$, $a_k$ on the real axis and the third point $a_i$ in the upper half plane.
The second relation follows from the fact that  the l.h.s. of the second relation equals to $\frac {1}{2}\int _C \sqrt{\frac{b-t}{(t-a_1)(t-a_2)(t-a_3)}} dt$ where $C$ is any circle bounding a disk containing  the triangle $\Delta_Q$ inside. To show that this integral equals  $\pi $ consider the limit when the radius of $C$ tends to infinity.
Similar considerations settle the third relation.
\end{proof}

\begin{lemma} \label{lem:lb}
Let $b$ be a point in the triangle $\Delta_Q$ and let $l_b$ be the straight line  passing  through $b$ and  parallel to the side $\overline{a_ja_k}$.
Then there exists a unique point $b'$ on $l_b$ such that 
\begin{align}
\int _{a_j}^{a_k}  \sqrt{\frac{b'-t}{(t-a_1)(t-a_2)(t-a_3)}} dt \in \Rea , \label{eq:intajakRea}
\end{align}
Moreover,  $b' \in \Delta_Q \cap l_b$.
\end{lemma}
\begin{proof}
By (\ref{eq:invaffinetr})  we can wlog assume that  the points $a_j$, $a_k$ lie on the real axis, $a_j <a_k$ and the point $a_i$ lies in the upper half plane. Let us show that the imaginary part of $f_{j,k}(b)$ is a monotone decreasing function when $b$ runs from left to right along   the line $l_b$. 
Take $c_1, c_2 \in l_b$ such that $\Re c_1 <\Re c_2 $, let $t$ be a real number such that $a_j <t<a_k$. Decomposing them into the real and imaginary parts  $c_1=c_1^r +\sqrt{-1} c_1^i$, $c_2=c_2^r +\sqrt{-1} c_2^i$, $a_i=d^r +\sqrt{-1} d^i$  we get  $c_1^r -t < c_2^r -t$, $d ^i >0$  and $c_1^i=c_2^i$.
Since $1/((a_k-t)(t-a_j))>0$ and
\begin{align}
\frac{c_l-t}{a_i-t}
= \frac{(c_l^r -t) (d^r -t) +c^i_ld^i+\sqrt{-1}(c^i_l (d^r-t) -(c^r_l -t)  d^i)}{(d^r-t)^2+(d^i)^2} \quad (l=1,2),
\end{align}
we get 
$$\Im \left[\frac{c_1-t }{ (a_i-t)(a_k-t)(t-a_j)}\right ] > \Im \left[\frac{c_2-t }{ (a_i-t)(a_k-t)(t-a_j)}\right ].$$
Thus 
$$\Im \sqrt{\frac{c_1-t )}{ (a_i-t)(a_k-t)(t-a_j)}} > \Im \sqrt{\frac{c_2-t }{ (a_i-t)(a_k-t)(t-a_j)}}$$ 
and
\begin{align}
\Im \int _{a_j}^{a_k}  \sqrt{\frac{c_1-t}{(t-a_1)(t-a_2)(t-a_3)}} dt > \Im \int _{a_j}^{a_k}  \sqrt{\frac{c_2-t}{(t-a_1)(t-a_2)(t-a_3)}} dt ,
\end{align}
proving the required monotonicity. 
Notice that  for any  $a_j <t<a_k$ the imaginary part of $(c-t)/(a_i-t)$ is always  positive if $c \in l_b$ is to the left of  $\Delta_Q$, and negative  if $c \in l_b$ is to the right of  $\Delta_Q$.
Hence condition (\ref{eq:intajakRea}) can not hold if $b' \not \in \Delta_Q$.
The results follows by the mean value theorem.
\end{proof}

\begin{rem}
Thus the three curves $\ga_1,\ga_2,\ga_3$ determined by (\ref{eq:ba123Rea}) have to intersect  the triangle $\Delta_Q$. 
Relation (\ref{eq:threeint}) implies that if two of these curves meet at a certain point then the third curve also 
passes through the same   point.
By Lemma \ref{lem:lb} any two of these curves meet at (at least) one point.
\end{rem}

\begin{lemma}\label{lm:inter}
The curves $\ga_1,\ga_2,\ga_3$ determined by (\ref{eq:ba123Rea}) meet   at exactly  one point which lies inside $\Delta_Q$.
\end{lemma}
\begin{proof}
If $\Delta_Q$ is an equilateral triangle, then $\ga_i$ is the straight line which passes through $a_i$ and is perpendicular to the side $\overline{a_ja_k}$.
Assume that for some $\Delta_Q$,  two curves $\ga_i$ and $\ga_j$ meet at more than one point.
Deform this $\Delta_Q$ into the equilateral triangle. During this deformation  these two curves experience a deformation during which they should touch each other tangentially. 
But this contradicts to Lemma~\ref{lm:trans}. 
\end{proof}
\begin{nota}
Let $b_0$ denote the point where $\ga_i,\ga_j,\ga_k$ meet.
Recall that we denote the segment of $\ga_i$ connecting  $a_i$ and $b_0$ by $\Ga_i$. 
Let $D _i$ be the domain bounded by $\Ga_j$, $\Ga _k$ and by the side $\overline{a_ja_k}$, see Fig.5. 
\end{nota}
\begin{center}
\begin{picture}(400,120)(0,0)
\put(100,20){\line(1,0){180}}
\put(100,20){\line(3,2){135}}
\qbezier(235,110)(257,65)(280,20)
\put(210,55){\circle*{3}}
\put(206,44){$b_0$}
\qbezier(100,20)(167,43)(210,55)
\qbezier(280,20)(257,40)(210,55)
\qbezier(235,110)(220,85)(210,55)
\put(90,15){$a_1$}
\put(282,15){$a_2$}
\put(238,110){$a_3$}
\put(207,80){$\Ga _3$}
\put(237,35){$\Ga _2$}
\put(157,33){$\Ga _1$}
\put(230,65){$D_1$}
\put(178,60){$D_2$}
\put(185,28){$D_3$}
\put(140,0){Figure. 5. Our notation}
\end{picture}\label{fg:T1}

\end{center}
\bigskip
Consider  the quadratic differential $\Psi=R(z)dz^2,$ where  
$R(z)=\frac{b-z}{(z-a_1)(z-a_2)(z-a_3)}$.  Take its (horizontal)  trajectory, i.e. a level curve: 
\begin{align}
\Im  \int _{z_0}^z  \sqrt{R(t)} dt =const,
\end{align}
 where $z_0$ is some fixed point. Assume  that $b$ is located inside $\Delta_Q$ where as above  $Q(z)=(z-a_1)(z-a_2)(z-a_3)$.   If $R(z)$ is non-vanishing  and regular at some $z=z^*$, then the curve $\frak H: \Im  \int _{z_0}^z  \sqrt{R(t)} dt =const$ passing through $z^*$ is analytic in a neighbourhood of $z^*$, and the tangential direction to $\frak H$ at $z^*$ is given by $(\Re \sqrt{R(z^*)}, -\Im \sqrt{R(z^*)})$. To see this note that   locally near $z^*$ one has
\begin{equation}
 \int _{z^*}^z \sqrt{R(t)}dt \sim \sqrt{R(z^*)} (z-z^*) +O((z-z^*)^2). 
\end{equation}
Analogously,  the vertical trajectory $\frak V$ of $\Psi$ (which is given by  $\Re  \int _{z_0}^z  \sqrt{R(t)} dt =const$) passing through $z^*$ is also analytic in a neighbourhood of $z^*$, and its tangential direction at $z^*$ is given by $(\Im \sqrt{R(z^*)}, \Re \sqrt{R(z^*)}$.
Note that the orientation of $\frak H$ and $\frak V$ depends on the choice of a branch of $\sqrt{R(z)}$.

\begin{nota}
For a fixed $z^* \in \Cplx \setminus \{ a_1,a_2,a_3,b\}$, denote by  $\theta _1$, $\theta _2$, $\theta _3$, and  $\phi $   the arguments of the complex numbers $a_1-z^*$, $a_2-z^*$, $a_3-z^*$, and $b-z^*$ resp. 
Let $\theta _{j' j}$, $\phi _{j}$ be the arguments of $a_{j'}-a_{j} $, $b- a_{j}$. Finally, let  $\tilde{\phi }_{i}$ be the argument of $a_{i}-b$.
\end{nota}

The above formula for the tangent direction implies the following statement. 

\begin{lemma}ÊIn the above notation 
\begin{enumerate}
\item
one singular horizontal trajectory emanates  from each simple pole $z=a_j $ of $R(z)$; its tangent direction is given by $\theta _{ij} +\theta _{kj} -\phi _j$ and points inside $\Delta_Q$;

\item one singular vertical trajectory emanates from each simple pole $z=a_i$ of $R(z)$; its tangent direction is given by $\theta _{ij} +\theta _{kj} -\phi _j+\pi $ and  points outside $\Delta_Q$;

\item three singular horizontal trajectories  emanate from a simple zero $z=b$  and their tangent directions are  given by $(\tilde{\phi }_1 +\tilde{\phi }_2 +\tilde{\phi }_3 +\pi (1+2m))/3$ $(m=0,1,2)$. 
\end{enumerate} 
\end{lemma}

\begin{proof} The only thing we need to check is that 
if $b $ is in $\Delta_Q$, then we have $\min (\theta _{ij}, \theta _{kj}) \leq \phi _j \leq \max (\theta _{ij}, \theta _{kj}) $ and $\min (\theta _{ij}, \theta _{kj}) \leq \theta _{ij} + \theta _{kj} -\phi _j \leq \max (\theta _{ij}, \theta _{kj})$
Hence the tangent direction of the horizontal  trajectory emanating from  any pole points inside $\Delta_Q$.
\end{proof}

\begin{prop} $ $ 
\\
(i) If $b\notin \Delta_Q$ and a quadratic differential   $\Psi=\frac{(b-z)dz^2}{(z-a_1)(z-a_2)(z-a_3)}$ is Strebel then at least one of its singular horizontal trajectories   emanating from its poles $a_1,a_2,a_3$ leaves $\Delta_Q$. 
\\
(ii) If $b \in \Gamma _i \setminus \{  b_0 \} $ $(i=1,2,3)$, then $a_j$ and $a_k$ are connected by a singular horizontal trajectory $\gamma$ and it does not contain the  point $z=b$.
\\
(iii) If $b \in D_j$, then the singular horizontal trajectory $\gamma '$ which starts at $a_k$ (resp. $a_i$),  goes inside the triangle $\Delta_Q$, crosses the side $\overline{a_ka_i}$ (resp. $\overline{a_ka_j}$), and leaves  the triangle $\Delta_Q$.
\end{prop}
\begin{center}
\begin{picture}(440,130)(0,0)
\put(00,40){\line(1,0){180}}
\put(00,40){\line(3,2){135}}
\qbezier(135,130)(157,85)(180,40)
\put(125,95){\circle*{3}}
\qbezier(125,95)(132,115)(135,130)
\put(121,84){$b$}
\qbezier(00,40)(120,80)(180,40)
\put(25,0){Figure 6. Case $b \in \Gamma _3$}
\put(205,40){\line(1,0){180}}
\put(205,40){\line(3,2){135}}
\qbezier(340,130)(362,85)(385,40)
\put(343,90){\circle*{3}}
\put(339,79){$b$}
\qbezier(340,130)(335,90)(235,80)
\qbezier(385,40)(295,60)(215,25)
\put(250,5){Case $b \in D_1$}
\end{picture}
\end{center}
\begin{proof}
Part  (i) is completely obvious. Indeed, since $K_\Psi$ is compact there should be a singular trajectory connecting  one of the poles to the zero $b$. Since $b$ is located outside $\Delta_Q$ the result follows. 

To  prove (ii) note that if $b$ coincides with $a_i$, then the  horizontal trajectory emanating from $a_j$ is the straight segment $\overline{a_ja_k}$. Now take $b \in \Ga _i$  sufficiently close to $a_i$.
Then the horizontal trajectory emanating from $a_k$ passes  close to $a_j$, because the direction of the horizontal trajectory changes continuously with $b$ unless the horizontal trajectory hits a singular point.
Assume that the horizontal trajectory emanating from $a_k$ does not pass through the point $a_j$.
Let $\theta _{j' j}$, $\phi _{j}$, $\theta $ be the arguments of $a_{j'}-a_{j} $, $b- a_{j}$, $z-a_j$ resp.
If $z$ is sufficiently close to $a_j$, then the direction of the horizontal trajectory is approximately given by $(\theta +\theta _{ij} +\theta _{kj} -\phi _j)/2$, and it follows from elementary affine geometry that the horizontal trajectories are approximately parabolas whose focus is $a_j$ and the angle of the axis of the symmetry is $\theta _{ij} +\theta _{kj} -\phi _j$, see Fig. 7. 
\begin{center}
\begin{picture}(240,140)(0,0)
\multiput(10,50)(-15,-10){3}{\line(-3,-2){10}}
\put (10,20){\vector(-4,1){30}}
\put (20,17) {vertical trajectory from $a_j$}
\put(10,50){\line(3,2){45}}
\put(10,50){\line(1,0){210}}
\put(10,50){\line(1,2){45}}
\qbezier(25,50)(25,55)(23,58)
\put(29,53){$\theta _{ij} \!+ \!\theta _{kj} \! -\! \phi _j$}
\put (80,85){\vector(-3,-1){30}}
\put (75,87) {horizontal trajectory from $a_j$}
\qbezier(140,70)(160,80)(180,76)
\qbezier(180,76)(210,70)(220,50)
\qbezier(140,70)(20,00)(-10,40)
\qbezier(-10,90)(-20,60)(-10,40)
\put(-19,35){$d$}
\put(1,56){$a_j$}
\put(215,40){$a_k$}
\put (165,65){\vector(-1,0){30}}
\put (170,62) {horizontal trajectory from $a_k$}
\put(00,00){Figure 7. Behavior of horizontal and vertical trajectories near $a_j$.}
\end{picture}
\end{center}
Hence the horizontal trajectory emanating from $a_k$ goes around the point $a_j$, and intersects  the vertical trajectory  emanating from $a_j$, see Fig. 7.  Denote their intersection point by $d$ and 
consider the integral
\begin{align}
& \left( \int _{a_j}^{d}  +\int _{d}^{a_k}  \right) \sqrt{\frac{b-t}{(t-a_1)(t-a_2)(t-a_3)}} dt  ,  \label{eq:intStokesantiStokes}
\end{align}
where the path from $a_j$ to $d$ is taken along the vertical trajectory from $a_j$ and the path from $d$ to $a_k$ is taken along the horizontal trajectory from $a_k$. 
By definition of the horizontal and  vertical  trajectories, the value of the integral from $d$ to $a_k$ is real and the one from $a_j$ to $d$ is pure imaginary.
Since the integration path does not hit a singular point, the imaginary part (resp. the real part) varies monotonely as the integration variable passes  the vertical  (resp. horizontal) trajectories.
Hence the imaginary part of the value of (\ref{eq:intStokesantiStokes}) is not zero, but this contradicts to the definition of $\Gamma _i$.

Therefore we obtain that if $b \in \Gamma _i$ and $b$ is sufficiently close to $a_i$, then $a_j$ and $a_k$ are connected by a smooth horizontal trajectory $\gamma$ which does not hit  the singular  point $z=b$. Let us move the point  $b$  along  $\Gamma _i$ away from $a_i$. 
Then $a_j$ and $a_k$ are still connected by a smooth horizontal trajectory as long  as the horizontal trajectory connecting them  does not hit  the singular point $z=b$. On the other hand, 
 if the horizontal trajectory  passes through $z=b$, then  $a_j, a_k$ and $b$ will be connected by a horizontal trajectory  and the integrals: 
 \begin{align}
& \int _{a_j}^{b} \sqrt{\frac{b-t}{(t-a_1)(t-a_2)(t-a_3)}} dt  , \quad \int _{a_i}^{a_k} \sqrt{\frac{b-t}{(t-a_1)(t-a_2)(t-a_3)}} dt\;  \label{eq:intStokesTP}
\end{align}
attain real values. Hence the point $b$ is also contained in $\Gamma _j$ and we conclude that  $b=b_0$.
Therefore, we have shown that  $a_j$ and $a_k$ are connected by a smooth horizontal trajectory  in case $b \in \Gamma _i$.

To prove  (iii) choose  $\tilde{b} \in \Gamma _i \setminus \{b_0 \}$, and let $l_{\tilde{b}}$ be the straight line which passes through $\tilde{b}$ and is parallel to the side $\overline{a_ja_k}$.
 Eq.(\ref{eq:invaffinetr}) implies that we can wlog assume that the points $a_j$, $a_k$ lie on the real axis, $a_j <a_k$ and the point $a_i$ lies in the upper half plane. 
Let $b \in l_{\tilde{b}}\cap \Delta_Q$ such that $\Re \tilde{b} < \Re b $.
Then   $b \in D_j$ by Lemma \ref{lem:lb}.
\begin{center}
\begin{picture}(400,110)(0,0)
\put(100,30){\line(1,0){180}}
\put(100,30){\line(3,2){135}}
\qbezier(235,120)(257,75)(280,30)
\qbezier(100,30)(167,53)(210,65)
\qbezier(280,30)(257,50)(210,65)
\qbezier(235,120)(220,95)(210,65)
\put(90,25){$a_j$}
\put(282,25){$a_k$}
\put(241,121){$a_i$}
\put(200,75){$\Gamma _i$}
\put(222,95){\circle*{3}}
\put(216,98){$\tilde{b}$}
\put(150,95){\line(1,0){135}}
\put(284,98){$l_{\tilde{b}}$}
\put(234,95){\circle*{3}}
\put(232,98){$b$}
\qbezier(100,30)(220,70)(280,30)
\put(230,70){$D_j$}
\put(178,75){$D_k$}
\thicklines
\qbezier(140,10)(220,60)(280,30)
\put(20,0){Figure 8. Why singular trajectories emanating from poles leave $\De_Q$}

\end{picture}
\end{center}
Let $z$ be a point such that $\Im z <\Im b $. By  comparing the angles of $b-z$ and $\tilde{b}-z$ one can easily conclude that  the horizontal trajectory  emanating from $a_k$ for the  quadratic differential $\Psi=\frac{b-z}{Q(z)}dz^2$  is located under the similar horizontal trajectory for   $\widetilde\Psi=\frac{\tilde b-z}{Q(z)}dz^2$, see Fig. 8. 
Since the horizontal trajectory   emanating from $a_k$ for $\widetilde \Psi$  hits the point $a_j$, one has that the horizontal trajectory emanating from $a_k$   for $\Psi$ must intersect the side $\overline{a_ja_k}$ and leave the triangle $\Delta_Q$. The intersection point of this trajectory with the side $\overline{a_ja_k}$ can not coincide with $a_j$ since in that case the integral (\ref{eq:ba123Rea}) will be  real which  contradicts  Lemma \ref{lem:lb}.
Let us vary   $b$ in $D_j$.
The horizontal trajectory  intersects the side $\overline{a_ja_k}$ as long as the intersection point does not coincide with $a_j$ or $a_k$, or the horizontal trajectory  hits the singular point $b$.
Hence  (iii) is settled  for any  $b \in D_i$.
\end{proof}

\begin{rem} One can show that part (i) of the latter Proposition holds independently of whether $\Psi$ is Strebel or not but we do not need this fact. 
\end{rem} 

\begin{corollary} 
If $b \in \Gamma _i \setminus \{  b_0 \} $ $(i=1,2,3)$, then $a_i$ and $b$ are connected by a smooth horizontal trajectory. Another horizontal trajectory  emanating from $b$ surrounds the trajectory  connecting $a_j$ and $a_k$ and returns to $b$, see Fig. 9.
\end{corollary}
\begin{center}
\begin{picture}(260,130)(0,0)
\put(30,40){\line(1,0){180}}
\put(30,40){\line(3,2){135}}
\qbezier(165,130)(187,85)(210,40)
\put(155,95){\circle*{3}}
\put(151,84){$b$}
\qbezier(30,40)(150,80)(210,40)
\qbezier(165,130)(157,115)(155,95)
\qbezier(155,95)(80,80)(60,70)
\qbezier(60,70)(10,55)(15,30)
\qbezier(15,30)(20,15)(150,15)
\qbezier(150,15)(220,15)(230,40)
\qbezier(155,95)(230,65)(230,40)
\put(30,30){$a_j$}
\put(200,30){$a_k$}
\put(170,128){$a_i$}
\put(33,0){Figure 9. All singular trajectories for $b \in \Gamma _3$}
\end{picture}
\end{center}

\section{Proving Theorem~\ref{takemura}}\label{s:takemura}

For this proof we need Theorem~\ref{ShTa} below whose formulation uses the definition of  the following measures. Take as  above $Q(z)=(z-a_1)(z-a_2)(z-a_3)$. Choose one of three roots $a_i$ and shift the variable $z=z-a_i$. (Abusing our notation we use the same letter for the shifted variable.) Then in this new coordinate one has $Q(z)=z^3+v_iz^2+w_i,\; i=1,2,3$. Define the functions $\xi_i(\tau)=-v_i(1-\tau)^2, \; \tau\in [0,1]$ and $\psi_i(\tau)=-w_i(1-(1-\tau)^2)(1-\tau)^2, \; \tau\in [0,1]$. Let $\omega_i(\tau),\; i=1,2,3.$ be the arcsine measure supported on the interval $\left[\xi_i(\tau)-2\sqrt{\psi_i(\tau)}, \xi_i(\tau)+2\sqrt{\psi_i(\tau)}\right]$ in the complex plane. Finally define the measure $M_i,\;i=1,2,3$ by averaging 
$$M_i=\int_0^1\omega_i(\tau)d\tau.$$ 
Results of \cite {ShT}Ê claim that each $M_i$ is supported on an ellipse uniquely determined by the triple of roots  $a_1,a_2,a_3$ with the root $a_i$ playing a special role. Moreover all these three measures have the property that their Cauchy transforms satisfy outside their respective supports one and the same linear inhomogeneous differential equation:
\begin{equation}\label{eq:CT}
Q(z)\C^{\prime\prime}(z)+Q'(z)\C'(z)+\frac{Q^{\prime\prime}(z)}{8}\C(z)+\frac{Q^{\prime\prime\prime}(z)}{24}=0.
\end{equation}

Recall that $\mu_n$ denotes the root-couting measure of the spectral polynomial $Sp_n(\lambda)$, see Introduction. The weak limit of the sequence $\{\mu_n\}$  (if it exists) is denoted by $\mu$. In these terms the main result of \cite {ShT} Êis as follows. 

\begin{theorem}\label{ShTa} If in the above notation the measure $\mu$ exists then each of the measures 
$M_i,\; i\in \{1, 2, 3\}$ have $\mu$ as its inverse balayage, i.e. $\mu$ and $M_i$ have the same logarithmic potential near infinity and the support of $\mu$ is contained inside the support of $M_i$. 
\end{theorem}

In fact the proof of Theorem~\ref{ShTa} in \cite {ShT} (as well as of the original Theorem 1.4 of \cite {KvA}) Ê extends without changing a single word in it to converging {\bf subsequences} of the original sequence  $\{\mu_n\}$. 

Thus any two converging subsequences of measures from $\{\mu_n\}$ have the same limiting logarithmic  potential near infinity. But notice additionally that the support of these measure must necessarily belong to $\Gamma_Q$ which is the main result of \S~4 above. Thus  the limiting measures have the same logarithmic potential in the complement $\bC\setminus \Ga_Q$. But then they should coincide since both of them are $\bar z$-derivative of the same function.  

Let us now prove  Theorem~\ref{takemura}. We show first that the whole sequence $\{\mu_n\}$  of root-counting measures for the whole sequence  of $\{Sp_n(\lambda)\}$ converges. This argument resembles that of at the very end of \S~3. Indeed, by part (2) of Theorem~\ref{th:my} for any $\epsilon>0$ $\exists N_\epsilon$ such that  for all $n\ge N_\epsilon$ all roots of all $Sp_n(\lambda)$ lie in the $\epsilon$-neighborhood of $\Gamma_Q$. Therefore,  by compactness the sequence $\{\mu_n\}$  contains a lot of (weakly) converging subsequences. Theorem~\ref{ShTa} and the argument following it  show  that any two of such converging subsequences have the same (weak) limiting measure which we denote by $\mu$. Let us show that then the whole sequence $\{\mu_n\}$ is converging to the same $\mu$. Indeed, assume 
that  $\{\mu_n\}$ is not converging to  $\mu$. Then we can find a subsequence $N'$ of the natural numbers such that $\mu_n$ stays away from a fixed neighbourhood of $\mu$ in the weak topology for all $n\in N'$. Again by compactness we can fins a subsequence $N$ of $N'$ such that the limit of $\{\mu_n\}$ over $N$ exists and is equal to $\mu$ by the above argument. But this contradicts to the assumptions that $\mu_n$ stays away from $\mu$ for all $n\in N'$. 

To show that the whole $\Gamma_Q$ must be the support of the limiting measure $\mu$ notice that the Cauchy transform of $\mu$ satisfies \eqref{eq:CT} whose only singularities are $a_1,a_2,a_3$ and $\infty$. One can check that the unique  solution of \eqref{eq:CT} with the asymptotics $\frac{1}{z}$ near infinity has a nontrivial monodromy at each singularity 
 $a_1,a_2,a_3$. Notice that the Cauchy transform of $\mu$ coincides with this solution extended from infinity to the whole 
 $\bC\setminus \Ga_Q$. But then the density of $\mu$ which is the $\bar z$-derivative of this solution  restricted to  $\bC\setminus \Ga_Q$ can not vanish at any generic point of $\Ga_Q$, i.e. outside of $b_0$. \qed

\section {Final remarks}

\noindent 
{\bf 1.} A generalized Lam\'e equation has the  form:
$$Q(z)S^{\prime\prime}(z)+P(z)S'(z)+V(z)S(z)=0,$$
where $\deg Q(z)=l\ge 2$, $\deg P(z)\le l-1$, and $\deg V(z)\le l-2$. Fixing $Q(z)$ and $P(z)$ one looks for $V(z)$ of degree at most $l-2$ such that the latter equation has a polynomial solution $S(z)$ of a given degree $n$, see many details in e.g. \cite {Sh}. Typically for a given Lam\'e equation and a given positive integer $n$ there exist $\binom {n+l-2}{n}$ such Van Vleck polynomials of degree $l-2$. Moreover,  they are exactly $\binom {n+l-2}{n}$ many for any given Lam\'e equation if they are counted with appropriate multiplicities and $n$ is sufficiently large. Interesting computer experiments can be found in e.g. \cite{ABM} Êand were also perform by the present authors. These experiments  lead us to the following conclusion.   Let $\mathcal V_n$ be the set (more exactly, a divisor) of all normalized Van Vleck polynomials (i.e monic polynomials proportional to Van Vleck polynomials) such that each of them has a Stieltjes polynomial $S(z)$ of degree exactly $n$ counted with their multiplicities. In fact, $\mathcal V_n$ can be interpreted as a finite probability measure in the space $Pol_{l-2}$ of all monic polynomials of degree $l-2$ if we assign to each polynomial in $\mathcal V_n$  a positive Dirac measure equal  to its multiplicity divided by $\binom {n+l-2}{n}$ .  The following is a (weaker) version of conjecture of the second author settled above. 

\begin{conj} The sequence $\{\mathcal V_n\}$ of finite measures converges to a probability measure 
$\mathcal V_Q$ in $Pol_{l-2}$ which depends  only on the leading coefficient $Q(z)$.
\end{conj}Ê
 
\noindent
{\bf 2.} Similar set-up was developed in \cite {Sh}Ê for linear differential operators of order exceeding $2$ of the form $\frak q=\sum_{i=1}^kQ_i(z)\frac{d^i}{dz^i}$ where $\deg Q_i(z)\le i$ and $\deg Q_k(z)=k$. This topic was continued in \cite {HSh} where it is shown that a very natural analog of the main object of the present paper, i.e. the quadratic differentials $\Psi=-\frac{\tilde V(z)}{Q(z)}dz^2$ appears for operators of higher order as well. It also has almost all closed trajectories which are continuous but only piecewise smooth, in general. One needs to develop a notion of a Strebel differential for order $ >2$ which (to the best of our knowledge) is a completely open problem at the moment. So we pose our question in minimal possible generality. 

\begin{pr} How to define a notion of a  rational cubic Strebel differential 
$$\Psi=\frac{U_1(z)}{U_2(z)}dz^3,$$
where $\deg U_1(z)=1$ and $\deg U_2(z)=4$.
\end{pr}Ê

\end{document}